\documentclass[10pt,journal,final,twocolumn]{IEEEtran}

\usepackage[T1]{fontenc}
\usepackage{graphicx}
\usepackage{amsthm}
\usepackage{amsfonts,amssymb,amsmath}
\usepackage{mathtools}
\usepackage{subcaption}
\usepackage{cite}
\usepackage{color}
\usepackage{algorithm, algorithmic}

\newtheorem{theorem}{Theorem}
\newtheorem{proposition}{Proposition}
\newtheorem{corollary}{Corollary}
\newtheorem{lemma}{Lemma}

\begin{document}

\title{Power Beacon Placement for Maximizing Guaranteed Coverage in Bistatic Backscatter Networks}

\author{Xiaolun~Jia,~\IEEEmembership{Graduate Student Member,~IEEE,}
        and~Xiangyun~Zhou,~\IEEEmembership{Senior Member,~IEEE}
        
\thanks{The authors are with the School of Engineering, The Australian National University, Canberra, ACT 2601, Australia (e-mail: \{xiaolun.jia, xiangyun.zhou\}@anu.edu.au).}
}

\maketitle

\begin{abstract}
The bistatic backscatter architecture, with its extended range, enables flexible deployment opportunities for backscatter devices. In this paper, we study the placement of power beacons (PBs) in bistatic backscatter networks to maximize the guaranteed coverage distance (GCD), defined as the distance from the reader within which backscatter devices are able to satisfy a given quality-of-service constraint. This work departs from conventional energy source placement problems by considering the performance of the additional backscatter link on top of the energy transfer link. We adopt and optimize a symmetric PB placement scheme to maximize the GCD. The optimal PB placement under this scheme is obtained using either analytically tractable expressions or an efficient algorithm. Numerical results provide useful insights into the impacts of various system parameters on the PB placement and the resulting GCD, plus the advantages of the adopted symmetric placement scheme over other benchmark schemes.
\end{abstract}

\begin{IEEEkeywords}
Bistatic backscatter communication, coverage, placement optimization, power beacons.
\end{IEEEkeywords}

\section{Introduction}

\IEEEPARstart{B}{ackscatter} communication has emerged as a promising technology to improve device lifetime. Rather than transmitting signals using active circuitry, communication is performed passively using existing radiofrequency (RF) signals. The use of simpler, passive circuitry as opposed to RF transmit chains significantly reduces the devices' power consumption. A prime example of backscatter communication is radiofrequency identification (RFID), in which a powered reader transmits an unmodulated continuous wave (CW) signal to be received by RFID tags. The tags modulate their data onto the CW by switching between load impedances connected to its antenna to vary the amount of reflected signal power, and the modulated signal returns to the reader to be decoded. Notable studies on RFID system performance can be found in \cite{Van08, Ble10, Boy12}.

Among the different backscatter architectures, the bistatic architecture, whose concept can be traced back to \cite{patent}, has received considerable recent interest due to its extended communication range. In a bistatic backscatter system, an RF carrier emitter, hereafter referred to as a power beacon (PB), is separated from the reader, such that the power-up and backscatter links undergo uncorrelated fading. By placing a backscatter device (BD) close to a PB, the backscattering range of the BD may be significantly improved from that in RFID systems. Link budgets for the bistatic setup were examined in \cite{GD09}. Initial experimental studies using commodity radios were conducted in \cite{Kim14}, where backscattering ranges of over $100$~m, or an order of magnitude further than in the monostatic and ambient architectures, were achieved. Work in \cite{Fas15} studied channel coding and interleaving at the BD to mitigate the effect of deep fades, and proposed various coherent detectors which enabled ranges of up to $150$ m. Work in \cite{Shen16} proposed solutions to the phase cancellation problem when unmodulated CWs are transmitted by PBs. Work in \cite{Alev17} extended \cite{Fas15} to account for noncoherent detection; while work in \cite{Zhu18} utilized the separation of the reader and PB to introduce a grant-free random access protocol for BDs to enable efficient concurrent communication over long distances with simplistic devices. Key applications envisioned to be fulfilled by bistatic backscatter systems include tracking, localization and environmental monitoring for smart cities, all of which fall under the sensor networks component of the Internet of Things paradigm, and place emphasis on extended coverage over large areas.

\subsection{Motivation and Contributions}

Much of the existing literature on bistatic backscatter has focused on performance studies centered around one representative BD in systems with one PB. Given the long ranges capable by BDs in bistatic configuration, it is natural to ask how the extended range can be translated to extended \textit{coverage} over larger areas to cater to the mentioned applications, which may comprise of many BDs. Specifically, the question of how the deployment of PBs, such as their number and locations, may be optimized in a cost-efficient way, remains to be answered. Few works have analyzed system-level performance of backscatter networks with multiple nodes in some way. Work in \cite{HH17} used stochastic geometry to model the placement of collocated readers and PBs and its effects on coverage and capacity in a network with many BDs. In studying the system-level performance of a bistatic backscatter network, work in \cite{Multistatic} provided some initial insights into the placement of multiple PBs by using a square grid placement scheme around the reader. However, the objective therein was not to optimize the PB placement to maximize coverage, and no further optimization insights were offered. Work in \cite{BC17} compared the coverage and capacity of a backscatter network and a non-backscatter wireless powered communication network (WPCN). However, PBs were placed randomly therein to facilitate D2D communication between BDs, as opposed to BD-reader communication. Thus, apart from \cite{Multistatic}, none of the other works provided design insights specifically applicable to the optimization of bistatic backscatter networks. Moreover, these works did not explore the placement optimization of PBs in bistatic backscatter networks for maximum coverage, which constitutes a significant knowledge gap.

An important consideration for the choice of network deployment strategy is to achieve good coverage with low cost. In bistatic backscatter systems, the cost of a reader can be one order of magnitude higher than the cost of a PB, as readers possess complex signal processing circuitry absent in PBs. Therefore, it is wise to deploy many PBs surrounding a small number of readers to achieve good coverage with low deployment cost. This has motivated the baseline setting considered in this paper, which is a bistatic backscatter network consisting of one reader and multiple PBs. We study the design problem involving deterministic PB placement in a bistatic setup with a central reader. Note that our scenario and design problem shares similarities with a cellular network, in which a much more expensive base station is placed at the center of the cell and a number of low-cost relay stations are deployed in the cell to improve communication performance such as coverage. Furthermore, it is reasonable to determine PB placement prior to BD deployment, so that the BDs may be moved from time to time. The chosen PB placement strategy should provide coverage over a targeted geographical area, rather than a specific set of BD locations. To the best of our knowledge, the PB placement optimization problem in bistatic backscatter networks is yet to be addressed in the literature.

We aim to carefully place PBs to maximize the distance from the reader within which BDs are guaranteed to meet a certain quality-of-service (QoS). That is, our PB placement design guarantees coverage for BDs based on their distances from the reader, instead of their specific locations. This allows the BD topology to be altered for application-specific needs, without requiring the PBs to be redeployed (for which significant effort is required). This is beneficial for the mentioned applications of bistatic backscatter, where large numbers of BDs are present over a certain geographical area, whose locations may be subject to change. The problem is unique in that the coverage areas of individual PBs are irregularly shaped, owing to the complete dependence of the BD-reader link on the PB-BD link under bistatic backscatter, a characteristic not present in any other similar wireless network. 

Our main contributions are as follows:
\begin{itemize}
\item We present the first study on the PB placement problem in bistatic backscatter networks. By defining the QoS in terms of an outage probability constraint, we formulate the PB placement problem for maximizing the guaranteed coverage distance (GCD), a metric which we propose for its straightforward characterization of the coverage area.

\item Due to the complex nature of each PB's coverage area, we adopt a circular symmetric placement strategy commonly used in WPCNs. We examine cases where each BD is served either by its nearest PB or simultaneously by multiple nearest PBs. Closed-form expressions for the optimal PB-reader distance and asymptotic GCD are derived for the case where each BD is served by its nearest PB; and a low-complexity algorithm is proposed to determine the optimal PB-reader distance when multiple serving PBs are involved.

\item Extensive numerical results are presented to characterize the impacts of system parameters on the GCD, and to highlight design insights. Particularly, a significant GCD improvement is observed by using two serving PBs for each BD instead of only one serving PB, but allowing more than two serving PBs quickly results in diminished returns. Comparisons with alternative PB placement schemes are also presented, which highlight the favorable performance of our placement scheme given its low-complexity nature.
\end{itemize}

\subsection{Related Work}

The PB placement problem studied in this paper shares certain common elements with existing literature on the relay placement problem in cellular networks and non-backscatter WPCNs, but presents unique differences in the context of bistatic backscatter and its properties.

\emph{1) Relay Placement Problem in Cellular Networks:} Some similarities exist between our PB placement problem and the relay placement problem in conventional communications \cite{P2, P4, P5, P8}, where the aim is often to determine optimal locations of powered relaying nodes to ensure that coverage, throughput or lifetime constraints are satisfied at distantly located nodes. As this problem is known to be intractable \cite{P2, P5}, works such as \cite{P8} proposed algorithms that guarantee solutions with various approximation ratios. However, our work differs significantly from conventional relay placement, wherein the source-relay and relay-destination links are decoupled due to the actively transmitting nature of relays, rendering their coverage areas as circular. The use of PBs in bistatic backscatter, where the BD-reader link completely depends on the PB-BD link, results in irregular coverage areas for each PB. This adds further complexity compared to the node placement problem compared to relay-assisted networks.

\emph{2) Wireless Powered Communication Networks:} In WPCNs, users harvest and store energy received from RF sources for active transmissions \cite{W9, W93, Guo19}. A number of studies on power source placement exist for WPCNs \cite{W7, W10, W12, W14, W94, W95, Ros20}. Specifically, \cite{W7} considered a symmetric deployment of PBs equidistant from a base station, and derived outage probability expressions in terms of the number of PBs and the PB-base station distance. A similar symmetric placement scheme was utilized in \cite{W95} to establish the optimal locations of the PBs' distributed antennas to maximize the wireless power transfer efficiency. Work in \cite{W10} assumed known user locations and proposed suboptimal algorithms to divide the network into partitions, with one PB deployed at a time to ensure that each PB and all PBs beforehand were placed in locally-optimal solutions. Despite the extensive literature, we note that these works have mostly studied scenarios where PB locations were either nonuniform to cover a specific user topology, or only catered to representative user locations in terms of metrics averaged over all users, as opposed to guaranteeing a certain level of performance for every user. As such, there still lacks placement strategies which achieve guaranteed, location-independent performance to users, and more so in backscatter networks.

We demonstrate that the PB placement scheme proposed in this paper is not only able to take advantage of the extended range of bistatic backscatter, but fulfills three key requirements highlighted by, and missing in, these prior works: cost-efficient deployment, straightforward characterization of irregular coverage areas, and guaranteed coverage for arbitrary BD topologies.

The rest of this paper is organized as follows. Section II introduces the system and signal models, and formulates the PB placement problem. Section III presents the solution to the PB placement problem for the case where each BD is served by its nearest PB. Section IV presents a low-complexity algorithm for solving the PB placement problem where each BD is served by multiple PBs. Numerical results are presented in Section V and Section VI concludes the paper.

\textit{Notations:} $j = \sqrt{-1}$ denotes the complex unit, and $|\cdot|$ denotes the magnitude of a complex number. $\mathbb{E}\left\{ \cdot \right\}$ and $\mathrm{Pr}(\cdot)$ denote the expectation operator and the probability of an event, respectively. $\Gamma(\kappa, \tau)$ denotes a gamma random variable with shape parameter $\kappa$ and scale parameter $\tau$. Boldface letters denote vectors, as in $\mathbf{a}$; $\left\lVert \mathbf{a} \right\rVert$ denotes the Euclidean norm of $\mathbf{a}$. 


\section{System Model and Problem Formulation}

We consider a bistatic backscatter network with one reader located at the origin and a number of BDs dispersed around the reader. The BDs are modeled as being stationary, but their locations within the network may change in the long term for application-specific purposes. The data packets generated at the BDs are to be transmitted to the reader via backscattering. We assume semi-passive BDs with built-in batteries to support their circuit operations, and extend the analysis numerically to passive BDs in Section V-B. Each BD has two load impedances connected to its antenna. Two modulation schemes are commonly used in bistatic backscatter networks: on-off keying (OOK) and frequency-shift keying (FSK). Where OOK is used, we assume that one BD is served at a time. That is, the BD being served needs to be acknowledged by the network through a contention process before transmission can occur. Where FSK is used, multiple BDs may simultaneously backscatter their data symbols using unique, sufficiently spaced subcarrier frequencies assigned to each BD. Since no interference occurs between the transmissions from the BDs in either case, we focus our analysis on a single BD. We show in the sequel that the coverage analysis for both modulation schemes can be performed using an identical procedure.

A total of $M$ PBs are deployed to facilitate the backscatter communication of the BDs. Each PB performs energy transmissions omnidirectionally \cite{W10, Ros20} with power $P$. We allow either one or multiple nearest PBs to serve a BD (e.g., two PBs can simultaneously provide energy transmissions to serve the same BD). The set of indices of PBs serving BD $k$ is denoted by $\Phi_{k}$.

For the rest of this paper, we abbreviate PBs (a.k.a. carrier emitters), BDs (a.k.a. tags) and the reader as $C$, $T$ and $R$, respectively, in boldface letters and subscripts. Indexing for PBs and BDs are $m \in \{0, ..., M-1\}$ and $k \in \{1, ..., K\}$. The system diagram is shown in Fig. \ref{fig:systemDiagram}.

\begin{figure}[!t]
\centering
\includegraphics[width=3.5in]{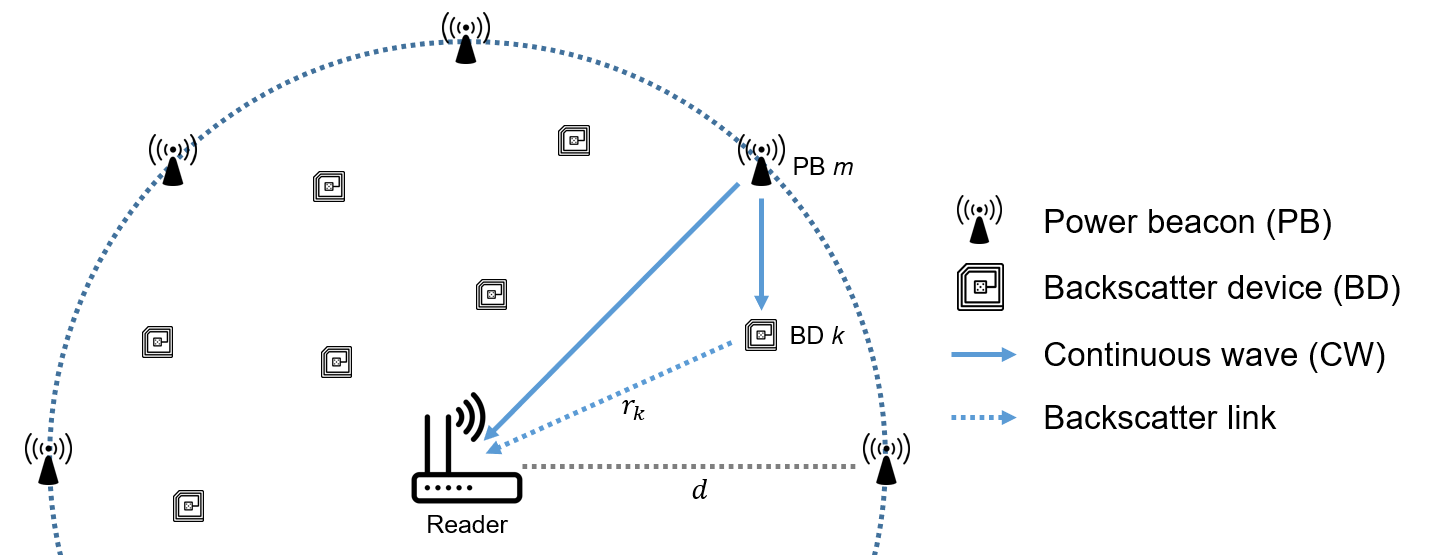}
\caption{Bistatic backscatter network with PBs.}
\label{fig:systemDiagram}
\end{figure}

\subsection{Channel and Signal Model}

Denote the complex channel coefficient from node $i$ to node $j$ by $h_{i,j}$, which accounts for small-scale fading such that $\mathbb{E}\left\{\left|h_{i,j}\right|^{2}\right\} = \beta_{0} \left|\left| \mathbf{n}_{i} - \mathbf{n}_{j} \right|\right|^{-\delta}$, where $\beta_{0} = \left(\frac{\lambda}{4\pi}\right)^{2}$ is the reference path loss; $\lambda$ is the carrier wavelength; $\delta \geq 2$ is the path loss exponent; and $\mathbf{n}_{k}$ is the coordinate vector of node $k$ of type $\mathbf{n} \in \{\mathbf{c},\mathbf{t},\mathbf{r}\}$ for PB, BD and reader, respectively. Frequency-flat quasi-static fading is assumed, consistent with existing backscatter works. In this paper, we adopt Nakagami small-scale fading due to its mathematical tractability, as in \cite{Multistatic, Ves18}, and also for its ability to approximate line-of-sight situations, which is common for the applications of bistatic backscatter considered in \cite{Kim14, Alev17, Multistatic}. We assume uniform propagation conditions throughout the network, noting that the nonuniform case (e.g., localized blockages) is another important problem, but is outside the scope of this work.

Each serving PB transmits an identical CW signal with power $P$. The equivalent baseband signal received at BD $k$ from the serving PB set is the superposition of all the CW components:
\begin{equation}
y_{k}(t) = \sqrt{P}\!\sum_{m \in \Phi_{k}} h_{C_{m},T_{k}} \ c(t), \label{eq:tagReceived}
\end{equation}
where $c(t)$ is the baseband representation of the CW. BD $k$ performs modulation on the CW by switching between its two impedances to generate two reflection coefficients, which determine the strength of the backscattered signal. The baseband information signal at BD $k$ is given by
\begin{equation}
b_{k}(t) = A_{k} - G_{k}(t), \label{eq:tagBaseband}
\end{equation}
where $A_{k} \in \mathbb{C}$ is the antenna structural mode of BD $k$, and $G_{k}(t) \in \mathbb{C}$ refers to the time-varying reflection coefficient function. $G_{k}(t)$ takes on two values $G_{k,0}$ and $G_{k,1}$, with unit magnitude or less. This results in two values for $b_{k}(t)$, denoted by $b_{k,0}$ and $b_{k,1}$. The backscattered signal is
\begin{equation}
x_{k}(t) = \sqrt{\eta} \ y_{k}(t) \ b_{k}(t), \label{eq:tagSent}
\end{equation}
where $\eta$ is the backscatter switching loss coefficient, herein modeled as a constant that is consistent across all BDs. Note that no noise term is present in $x_{k}(t)$, which is a common assumption in the backscatter literature, as BDs do not perform any signal processing. The reader receives the summation of the backscattered information signal and a direct-link CW signal, in addition to the noise. 

We first derive the SNR over one symbol period at the reader for a symbol transmitted by the BD under OOK. The received signal can be written as
\begin{align}
x_{R}(t) &= \sqrt{P}\!\sum_{m \in \Phi_{k}} h_{C_{m},R} \ c(t) \nonumber \\ 
& \quad + \sqrt{P \eta} \ h_{T_{k},R} \ b_{k}(t)\!\sum_{m \in \Phi_{k}} h_{C_{m},T_{k}} \ c(t) + w(t), \label{eq:readerReceived}
\end{align}
where $w(t)$ is the noise at the reader. We assume that the durations of the CW and backscatter transmissions are identical, and that these signals are received at the same time at the reader without delay \cite{Wang16}. As the baseband representation of the CW signal at the reader is essentially a DC term\footnote{In this paper, DC terms refer to signal components which may be considered as constant terms at baseband, following the convention in \cite{Kim14}.}, the first term in (\ref{eq:readerReceived}) is a constant term bearing no information. Thus, the first term in (\ref{eq:readerReceived}) may be subtracted from the overall received signal without affecting the difference between the two levels in the second term of (\ref{eq:readerReceived}). The second term represents another constant term from the backscattered CW plus an oscillating term between the two levels of $b_{k}(t)$ from the backscatter modulation.  Noting that the latter can be split into an average term of the two levels caused by $b_{k}(t)$ (which is another DC term) and a modulated component, we rewrite (\ref{eq:readerReceived}) as follows:
\begin{align}
\tilde{x}_{R}(t) &= \bigg( \sqrt{P \eta} \ h_{T_{k},R}\!\sum_{m \in \Phi_{k}} h_{C_{m},T_{k}} \ c(t) \bigg) \nonumber \\ & \quad \times \left( \frac{b_{k,0} + b_{k,1}}{2} + \frac{b_{k,0} - b_{k,1}}{2} \mathcal{B}(t) \right) + w(t), 	\label{eq:readerReceived_noDL}
\end{align}
where $\mathcal{B}(t) \in \{-1, 1\}$, representing the two levels. Note that $\frac{b_{k,0} + b_{k,1}}{2}$ is also a DC term and can be removed. Following the removal of the second DC term from the backscattered CW, the received signal at the reader is discretized with sampling period $T$, which results in two possible values for the received signal:
\begin{equation}
x_{R}[l]\!=\!\pm\!\bigg(\!\sqrt{P \eta} \ h_{T_{k},R}\!\sum_{m \in \Phi_{k}} h_{C_{m},T_{k}} c(lT)\!\bigg) \frac{b_{k,0}\!-\!b_{k,1}}{2} + w[l]. \label{eq:readerReceivedCases}
\end{equation}
We let each symbol period $T_{sym}$ be comprised of $L$ samples, i.e., $T_{sym} = L T$.

Under OOK, the reflection coefficient $b_{k}[l]$, in the form of $b_{k,0}$ and $b_{k,1}$ in (\ref{eq:readerReceivedCases}), remains constant over $L$ samples of one symbol. Let $p$ and $q$ represent the non-noise term in the two quantities in (\ref{eq:readerReceivedCases}), with $i = 0$ and $1$, respectively. That is, $p = -\left( \sqrt{P \eta} \ h_{T_{k},R} \ \sum_{m \in \Phi_{k}} h_{C_{m},T_{k}} \right) \frac{b_{k,0} - b_{k,1}}{2}$ and $q = \left( \sqrt{P \eta} \ h_{T_{k},R} \ \sum_{m \in \Phi_{k}} h_{C_{m},T_{k}} \right) \frac{b_{k,0} - b_{k,1}}{2}$. In defining the SNR, the `signal' term is given by the difference between $|q|$ and $|p|$. Setting $|p|$ as the reference level, we take $(|q| - |p|)^{2}$ as the signal energy and obtain the SNR expression over one symbol period in (\ref{eq:SNR}). This expression is similar compared to \cite[Eq. (39)]{Kim14}, except that the CW is always on during each transmission in our case, resulting in the removal of the $\frac{1}{2}$ constant therein:
\begin{align}
\gamma_{k} &= \frac{\left(|q| - |p|\right)^{2}}{N_{0}} L \nonumber \\
&= \frac{P \eta \left|h_{T_{k},R}\right|^{2} |b_{k,0} - b_{k,1}|^{2} \left| \sum_{m \in \Phi_{k}} h_{C_{m},T_{k}}\right|^{2}}{N_{0}} L, \label{eq:SNR}
\end{align}
where $N_{0}$ is the variance of the discrete-time noise samples $w[l]$. Note that (\ref{eq:SNR}) is applicable to any BD in the system, due to the fact that one BD is served at any given time without interference from other BD transmissions.

Now we consider FSK modulation. In this case, switching between the two signal levels, similar to those in (\ref{eq:readerReceivedCases}), occurs over each symbol period at one of two subcarrier frequencies denoted by $f_{k,0}$ and $f_{k,1}$ for the $k$-th BD, which represent bits $0$ and $1$, respectively. Thus, the discretized baseband signal at the reader, after removing the DC terms, is given as \cite[Eq. (52)]{Kim14}:
\begin{align}
x_{R}[l] &= \bigg( \sqrt{P \eta} \ h_{T_{k},R}\!\sum_{m \in \Phi_{k}} h_{C_{m},T_{k}} \bigg) \nonumber \\
& \quad \times \frac{b_{k,0} - b_{k,1}}{2} \ \textsf{SW}(i, lT) + w[l],    \label{eq:readerReceivedFSK}
\end{align}
where $\textsf{SW}(i, t) = \frac{4}{\pi} \cos(2 \pi f_{i} t + \phi)$ is the baseband representation for a square wave with $50$\% duty cycle, frequency $f_{i}$ and filtered at the fundamental frequency, with $\phi \in [0, 2\pi]$ and $i \in \{0, 1\}$. As the cosine term oscillates between $-1$ and $1$, the two signal levels, similar to $p$ and $q$ before, are given by $\pm \left( \sqrt{P \eta} \ h_{T_{k},R} \sum_{m \in \Phi_{k}} h_{C_{m},T_{k}} \right) \frac{2(b_{k,0} - b_{k,1})}{\pi}$. Notice that the received signal under FSK has an additional $\frac{2}{\pi}$ term compared to OOK, and that the two signal levels apply to both bits $0$ and $1$, with the two symbol representations differing only in the frequencies of their CW components at baseband. The signal is demodulated using a correlator demodulator set to frequencies $\pm f_{k,0}$ and $\pm f_{k,1}$, where, if bit $0$ is sent, the demodulator at $\pm f_{k,0}$ outputs magnitude equal to (\ref{eq:readerReceivedFSK}), while the demodulator at $\pm f_{k,1}$ outputs the noise power, and vice versa. The SNR over one symbol period is similarly derived by rearranging the corresponding SNR definition in \cite{Kim14} noting that the CW is always on during each transmission, and also applies to any BD in the system due to the orthogonal subcarrier frequencies assigned to each BD:
\begin{equation}
\gamma_{k} =  \frac{4 P \eta \ |h_{T_{k},R}|^{2} \ |b_{k,0} - b_{k,1}|^{2} \left|\sum_{m \in \Phi_{k}} h_{C_{m},T_{k}}\right|^{2}}{\pi^{2} N_{0}} L.	\label{eq:SNR_FSK}
\end{equation}
Thus, we find that the SNR definition under FSK takes on an equivalent form compared to that of OOK in (\ref{eq:SNR}), in that both expressions contain the common term $|b_{k,0} – b_{k,1}|^{2}$ denoting the difference between the two energy levels, and differing only by a multiplicative constant of $\frac{4}{\pi^{2}}$. This means that the SNR-based analysis and design for both OOK and FSK are essentially the same (with a slight difference in the interpretation of results due to the multiplicative constant). Thus, we will use (\ref{eq:SNR}) as the SNR definition in our analysis and design hereafter.

During system operation, the reader must establish knowledge of certain statistics required for detection through training procedures, which vary with the modulation scheme. For both OOK and FSK, the transmission of deterministic training sequences is required from the BD. For OOK, the reader is able to establish the detection threshold from only the energy levels, without requiring the knowledge of instantaneous channel coefficients. For FSK, channel coefficients are required, and may be estimated using the least-squares approach in \cite{Fas15}.

We note the existence of several additional parameters in the link budget for bistatic backscatter communication, as outlined in \cite{GD09}, including antenna gains, polarization mismatch and on-object gain penalties. These can be readily added to the numerator or denominator of the SNR for gains and losses, respectively.

\subsection{Performance Metric and Problem Formulation}

In this work, we aim to address the PB placement problem for maximizing the coverage area of the network. The metric we use to represent coverage is the GCD, denoted by $r_{cov}$. Given a reader at the center of the network, the GCD is defined as the maximum distance at which a BD can be located from the reader, such that it is able to achieve a given QoS requirement. In other words, any BD located at a distance smaller than $r_{cov}$ from the reader is guaranteed to achieve the QoS requirement. The QoS is indicated by the outage probability, given for BD $k$ by $P_{out,k} = \mathrm{Pr}\left(\gamma_{k} < \gamma_{th}\right)$, where $\gamma_{th}$ is the minimum acceptable SNR threshold. The QoS requirement is then given by a maximum acceptable outage probability, denoted by $\varepsilon$. In other words, if BD $k$ is located at a distance smaller than $r_{cov}$, it is guaranteed to have $P_{out,k} < \varepsilon$.

The problem considers the placement of $M$ PBs in carefully chosen locations to maximize the GCD. In general, finding the optimal PB placement is analytically intractable and typically requires exhaustive search. To facilitate the network design, we consider one placement scheme that deploys the PBs at equal distance $d$ from the reader, with equal angular separation, similar to prior works on wireless power transfer such as \cite{W7, W95, Ros20}. Different to these works, however, the additional backscatter link must also be accounted for in the coverage analysis. We numerically demonstrate the favorable performance of this placement strategy compared to other schemes in Section V-C. We aim to determine the optimal PB-reader distance, $d^{*}$ (hereafter referred to simply as `PB distance'), such that $r_{cov}$ is maximized. 

The considered scheme is independent of specific BD locations (i.e., the design of $d^{*}$ is not specifically optimized for any BD topology), such that BDs can be installed, removed or relocated without requiring the PBs to be redeployed. A further important consideration in our placement strategy is to guarantee no coverage holes, i.e., all locations within distance $r_{cov}$ of the reader must satisfy the QoS requirement, which renders random (as opposed to deterministic) placement schemes inappropriate for this purpose. Coverage holes occur when the PB-reader distance is large, causing the coverage area to become discontinuous. To avoid coverage holes, and to ensure ease of BD deployment without the need to manually verify coverage at every location and over a range of channel conditions over time, it is necessary to impose a QoS requirement at all locations. The use of a QoS requirement accounts for the long-term distribution of the fading channels. The GCD metric distinguishes our problem compared to e.g., \cite{W7}, which performs the placement optimization by averaging the outage probability over the entire network without guaranteeing a minimum QoS requirement at every location. The problem can be written as:
\begin{subequations}
\begin{align}
\text{(P1)}: ~~\max_{d} ~~~&r_{cov}(\{\mathbf{c}_{i}\}, \{\mathbf{t}_{k}\}) \label{eq:P1a} \\
\mathrm{s.t.}~~~~& P_{out,k} < \varepsilon, \forall k: \lVert \mathbf{t}_{k} \rVert < r_{cov}, \label{eq:P1b} \\
& \lVert \mathbf{c}_{i} \rVert = d, \forall i \in \{1, ..., M\}, \label{eq:P1c} \\
& \theta_{m_{i}, m_{j}} = \frac{2 \pi}{M}, \forall i, j: i \in \{0, ..., M-1\}, \nonumber \\
& \qquad \qquad j \in \{i-1, i+1\} \ \mathrm{mod} \ M, \label{eq:P1d}
\end{align}
\end{subequations}
where $\{\mathbf{c}_{i}\}$ and $\{\mathbf{t}_{k}\}$ denote the sets of PB and BD location vectors; and $\theta_{m_{i}, m_{j}} = \theta_{m_{i}} - \theta_{m_{j}}$ represents the angular spacing between adjacent PBs $i$ and $j$. Constraint (\ref{eq:P1b}) states that all BDs within the GCD must satisfy the QoS requirement; (\ref{eq:P1c}) requires that all PBs must be equidistant from the reader; and (\ref{eq:P1d}) requires equal angular spacing between adjacent PBs. (P1) is complicated by the fact that an expression for $r_{cov}$ is not available beforehand. To address this significantly limiting factor and to improve mathematical tractability, in the following sections, we solve an equivalent problem, and in the process, obtain insights into the behavior of $r_{cov}$.

Hereafter, we drop the BD indexing and focus on one representative BD. In the following sections, we will separately analyze the cases of single and multiple serving PBs.


\section{PB Placement Analysis: Single Serving PB}

In this section, we analytically derive the optimal PB placement distance when each BD is served only by its nearest PB. First, the exact outage probability expression for an arbitrarily located BD is presented, in terms of its distance from the reader $r$ (to refer to the communication \underline{r}ange of BDs), the PB \underline{d}istance $d$, and the number of PBs $M$. We then solve the PB placement problem based on our proposed placement scheme, and provide exact expressions for $d^{*}$. For ease of exposition, we denote all distance quantities associated with a PB's location by the variable $d$ and associated subscripts, and assign the variable $r$ and associated subscripts to describe the location for an arbitrary BD.

\subsection{Outage Probability of an Arbitrarily Located BD}

We first determine the outage probability of a BD, conditioned on its distance from the reader $r$, and its angle relative to the $x$-axis, $\theta$. Allowing each BD to be served only by its nearest PB is equivalent to each PB serving only the users located to either side within an angular distance of $\frac{\pi}{M}$. In other words, the PB's serving region is a sector with angular width $\frac{2 \pi}{M}$, with the PB located on the line bisecting the sector. Thus, all sectors are identical for the purpose of analysis, and we consider a representative sector (Fig. \ref{fig:contours}) in the following analysis.

\begin{figure*}
\centering
\includegraphics[width=0.7\textwidth]{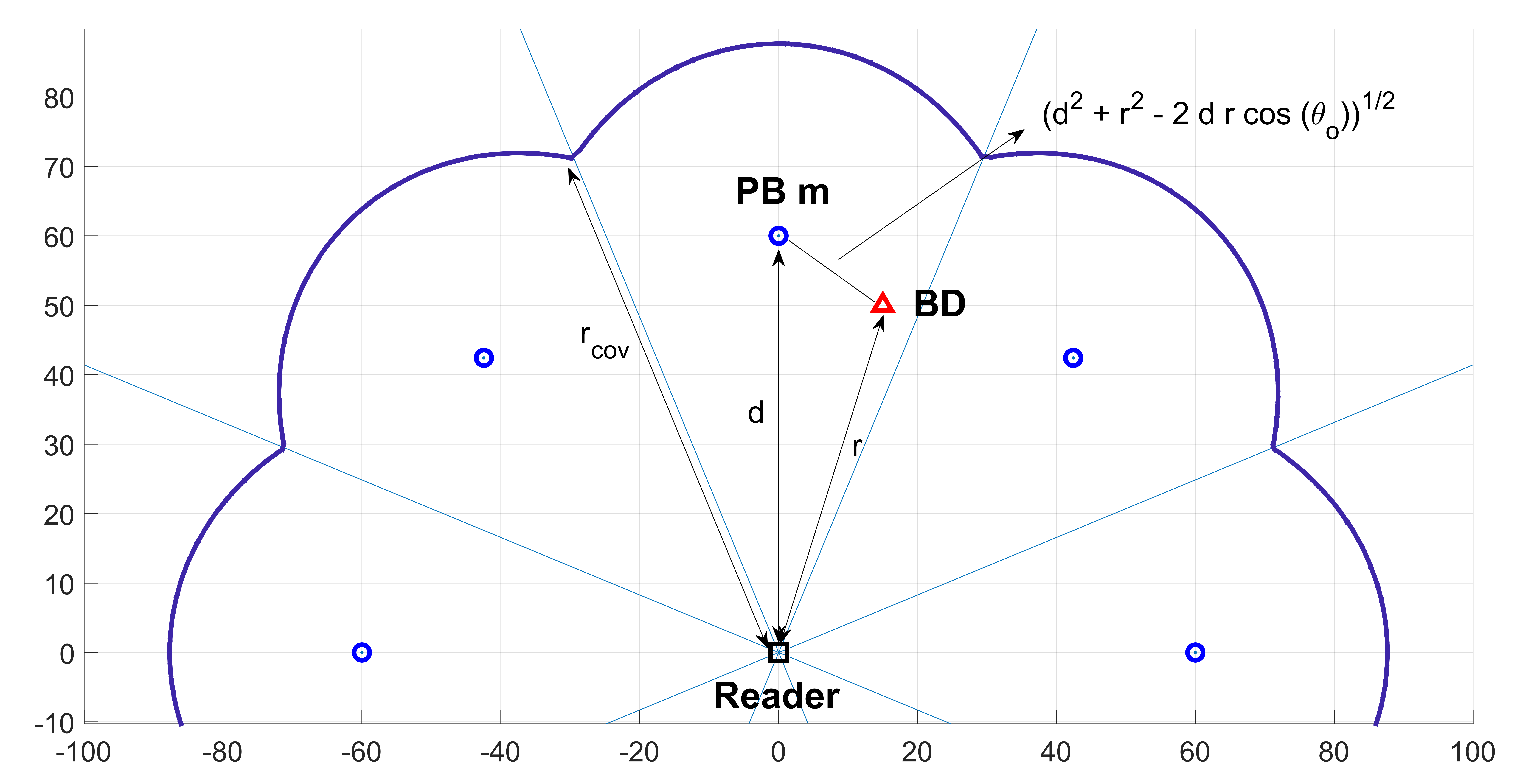}
\caption{Illustration of sectors for each serving PB. Also shown is an example contour within which a given QoS constraint is satisfied. The corresponding GCD is denoted by $r_{cov}$.}
\label{fig:contours}
\end{figure*}

For a BD located at $(r, \theta)$, the PB-BD and BD-reader link distances are given by $\sqrt{d^{2} + r^{2} - 2 d r \cos(\theta_{o})}$ and $r$, respectively, where $\theta_{o} = \theta - \theta_{m}$ denotes the angular offset between the BD and its nearest PB, with $\theta_{m}$ representing the angle relative to the $x$-axis of the nearest PB with index $m$. The per-symbol SNR at the reader can be written as:
\begin{equation}
\gamma = \gamma_{eq} X Y, \hspace{3mm} \gamma_{eq} = \frac{P \eta \beta_{0}^{2} L |b_{0} - b_{1}|^2}{r^{\delta} \left( d^{2} + r^{2} - 2 d r \cos(\theta_{o}) \right)^{\frac{\delta}{2}} N_{0}}, \label{eq:eqSNR_one}
\end{equation}
where $\gamma_{eq}$ can be understood as the equivalent SNR without fading, and $X, Y \sim \Gamma\left(\kappa, \frac{1}{\kappa}\right)$ are independent and identically distributed (i.i.d.) random variables of the squared magnitude of the PB-BD and BD-reader channel coefficients, respectively. Hereafter, the terms shape parameter and Nakagami parameter are used interchangeably in reference to gamma random variables.

The probability density function (pdf) and cumulative distribution function (cdf) of a gamma random variable, with shape parameter $\kappa$ and scale parameter $\frac{1}{\kappa}$, are given by
\begin{align}
f_{X}(x) &= \frac{\kappa^{\kappa} x^{\kappa-1}}{\Gamma(\kappa)} \exp(-\kappa x), \nonumber \\
F_{X}(x) &= 1 - \exp(-\kappa x) \sum_{n=0}^{\kappa-1} \frac{(\kappa x)^{n}}{n!}, \label{eq:gamma_pdf_cdf}
\end{align}
where, for ease of exposition, we have used the cdf expression corresponding to the case where $\kappa$ is an integer. The outage probability at the reader, $P_{out}$, is given by
\begin{align}
P_{out}	&= \mathrm{Pr}(\gamma_{eq} X Y < \gamma_{th}) \hspace{5mm} (= F_{XY}\left(\frac{\gamma_{th}}{\gamma_{eq}}\right)) \nonumber \\
		&= \int_{0}^{\infty} F_{X}\left( \frac{\gamma_{th}}{\gamma_{eq} y} \right) f_{Y}(y) \ \mathrm{d}y \nonumber \\
		&= 1 - \sum_{n=0}^{\kappa_{X}-1} \frac{2}{n! \Gamma(\kappa_{Y})} \left( \frac{\kappa_{X} \kappa_{Y} \gamma_{th}}{\gamma_{eq}} \right)^{\frac{\kappa_{Y} + n}{2}} \nonumber \\
		& \quad \times K_{\kappa_{Y} - n} \left( 2 \sqrt{\frac{\kappa_{X} \kappa_{Y} \gamma_{th}}{\gamma_{eq}}} \right), \label{eq:Pout_one}
\end{align}
where the last expression follows from \cite{Van16}; $\kappa_{X}$ and $\kappa_{Y}$ are the shape parameters of $X$ and $Y$, respectively; and $K_{\nu}(\cdot)$ is the modified Bessel function of the second kind with order $\nu$. We note that the outage probability expression under the case of non-integer $\kappa$ may be derived using the general cdf $F_{X}(x) = \frac{1}{\Gamma(\kappa)} \gamma(\kappa, \kappa x)$, where $\gamma(\cdot, \cdot)$ is the lower incomplete gamma function. However, as we show in the next subsection, we can bypass the complexity of solving Problem (P1) using the complete outage probability expression, by utilizing the concept of equivalent SNR first introduced in (\ref{eq:eqSNR_one}), which applies to any value of $\kappa$.

\subsection{Solution to the PB Placement Problem}

Due to the unwieldy form of $P_{out}$, it is difficult to solve (P1) directly. Thus, we transform this design problem into a purely geometric form, by rearranging the outage probability constraint into an equivalent SNR constraint:
\begin{subequations}
\begin{align}
\text{(P2)}: ~~\max_{d} ~~~&r_{cov} \label{eq:P2a} \\
\mathrm{s.t.}~~~~& \gamma_{eq} > \gamma_{eq,th} \ \mathrm{\&\&} \ \text{(\ref{eq:P1d})}, \label{eq:P2b}
\end{align}
\end{subequations}
where $\gamma_{eq,th}$ is the equivalent SNR threshold that is required to ensure the original outage probability constraint $P_{out} < \varepsilon$ is met. It can be given in terms of $\gamma_{th}$ as
\begin{equation}
\gamma_{eq,th} = \frac{\gamma_{th}}{F_{XY}^{-1}(\varepsilon)},  \label{eq:conv_one}
\end{equation}
where $F_{XY}^{-1}$ is the inverse of the cdf of the product of $X$ and $Y$ from rearranging (\ref{eq:Pout_one}). Unfortunately, no closed form expression exists for $F_{XY}^{-1}$; however, it is monotonic and can be empirically constructed from samples generated from the product of two gamma random variables.

Based on (P2), we present the following result on the characteristic of the coverage area, which significantly simplifies subsequent analysis. For clarity, we rewrite $\gamma_{eq}$ in (\ref{eq:eqSNR_one}) as 
\begin{equation}
\gamma_{eq} = \frac{\alpha}{r^{\delta} \left( d^{2} + r^{2} - 2 d r \cos(\theta_{o}) \right)^{\frac{\delta}{2}}},	\label{eq:eqSNR_one1}
\end{equation}
where $\alpha = \frac{P \eta \beta_{0}^{2} L |b_{0} - b_{1}|^2}{N_{0}}$ denotes the distance-independent terms in the SNR.

\begin{proposition}
For a BD deployed at distance $r$ from the reader, the minimum equivalent SNR is obtained when the BD is located at the edge of the sector, i.e., $\theta_{o} = \frac{\pi}{M}$.
\end{proposition}

\begin{proof}
Without loss of generality, we consider a PB located on the $x$-axis. Then we have
\begin{equation}
\frac{\mathrm{d}\gamma_{eq}}{\mathrm{d}\theta_{o}} = - \frac{\alpha d r^{3} \delta}{(d^{2} r^{2} + r^{4} - 2 d r^{3} \cos(\theta_{o}))^{1 + \frac{\delta}{2}}} \sin(\theta_{o}).	\label{eq:prop1}
\end{equation}
It is evident that all terms in (\ref{eq:prop1}) are positive for $\theta_{o} \in [0, \frac{\pi}{M}]$. For $\theta_{o} \in [-\frac{\pi}{M}, 0]$, we have $-\sin(\theta_{o}) \geq 0$. Hence, the derivative is negative in the former range and positive in the latter, and is thus an odd function passing through $0$ when $\theta_{o} = 0$. As a result, minima for $\gamma_{eq}$ must occur at the edges of the sector, i.e., $\theta_{o} = \pm \frac{\pi}{M}$.
\end{proof}

Having established that the worst-case location for a BD is on the sector edge, we present our main result in the following theorem.

\begin{theorem}
The optimal PB distance $d^{*}$ when each BD is served by its nearest PB is given by
\begin{equation}
d^{*}\!=\!\begin{cases}
		\Big( \frac{\varsigma}{\left( 1 - \cos^{2} \left( \frac{\pi}{M} \right) \right)} \Big)^{\frac{1}{4}} \cos \left( \frac{\pi}{M} \right), & M\!\leq\!\left\lfloor{M_{th}}\right\rfloor, \\
		\left( -\frac{1}{2} \csc^{2} \left( \frac{\pi}{M} \right)  \right)^{\frac{1}{4}} \nonumber \\
		\times \big( \sqrt{\varsigma^{2} \cos^{2}\!\left( \frac{\pi}{M} \right) \left( 9 \cos^{2}\!\left( \frac{\pi}{M} \right)\!-\!8 \right)^{3}} \nonumber \\ + \varsigma \left( 27 \cos^{4}\!\left( \frac{\pi}{M} \right)\!-\!36 \cos^{2}\!\left( \frac{\pi}{M} \right)\!+\!8 \right) \big)^{\frac{1}{4}}, & M\!\geq\!\left\lceil{M_{th}}\right\rceil,
	\end{cases}
\end{equation}
where $\left\lfloor{\cdot}\right\rfloor$ and $\left\lceil{\cdot}\right\rceil$ denote the floor and ceiling functions, respectively; $\varsigma = \left( \frac{\alpha}{\gamma_{eq,th}} \right)^{\frac{2}{\delta}}$; and $M_{th}$ is obtained at the crossover point of the two piecewise expressions above, given by
\begin{equation}
M_{th} = \frac{\pi}{\sec^{-1}(\omega)} \approx 12.36,	\label{eq:M_th}
\end{equation} 
where $\omega$ is the positive real root of the equation $4 x^6 + 2 x^2 - 7 = 0$.
\end{theorem}

\begin{proof}
See Appendix A.
\end{proof}

By straightforward evaluation of the first piecewise expression in Theorem 1, we find that $d^{*} = 0$ for $M = 2$. For $M = 1$, it is easy to see that any movement of the PB from the origin reduces the GCD, as the coverage area becomes irregular with respect to the origin, thus $d^{*} = 0$ also. Two additional results from Theorem 1 concern the asymptotic behavior of $d^{*}$ and $r_{cov}$ as $M$ becomes large, which are presented next.

\begin{corollary}
When $M \rightarrow \infty$, the optimal PB distance converges to
\begin{equation}
d^{*}_{\infty} = 2 \varsigma^{\frac{1}{4}}.  \label{eq:d_infty}
\end{equation}
Moreover, the corresponding GCD, $r_{cov}$, as $M \rightarrow \infty$, can be approximated by
\begin{multline}
r_{cov,\infty} \approx \varsigma^{\frac{1}{4}} + \frac{1}{\sqrt{6}} \bigg( \bigg(\frac{( \sqrt{\varsigma} + (-\varsigma^{\frac{3}{2}})^{\frac{1}{3}} )^{2}}{(-\varsigma^{\frac{3}{2}})^{\frac{1}{3}}}\bigg)^{\frac{1}{2}} \\ + \bigg(\frac{4 \varsigma + \frac{\varsigma^{2}}{(-\varsigma^{\frac{3}{2}})^{\frac{2}{3}}} + (-\varsigma^{\frac{3}{2}})^{\frac{2}{3}}}{\sqrt{\varsigma}}\bigg)^{\frac{1}{2}} \bigg).	\label{eq:corollary1}
\end{multline}
\end{corollary}
Equation (\ref{eq:d_infty}) is obtained from the expression of $d^{*}$ in Theorem 1 by letting $M \rightarrow \infty$.  Equation (\ref{eq:corollary1}) is obtained by solving (\ref{eq:f(r)}) in Appendix A, taking the expression of the root which represents the GCD (i.e., the smallest positive real root) with the value of $d$ set to that obtained in (\ref{eq:d_infty}), and then taking the limit as $M \rightarrow \infty$. Note that equation (\ref{eq:corollary1}) indicates an upper bound on the achievable guaranteed coverage distance for any value of $M$.


\section{PB Placement Analysis: Multiple Serving PBs}

In this section, we extend the previous analysis to the case where each BD is served by multiple nearest PBs. In Section II, we have used $\Phi$ to denote the set of indices of PBs serving a particular BD. Hereafter, we denote the number of nearest serving PBs as $S$ (i.e., the number of elements in $\Phi$), and consider $2 \leq S \leq M$ in this section. The outage probability expression for an arbitrarily located BD is presented first, followed by an algorithm for determining $d^{*}$.

\subsection{Outage Probability of an Arbitrarily Located BD}

The outage probability expression for the case where each BD performs backscatter transmission using the CW from $S$ nearest PBs is derived similarly compared to Section III-A. The received energy at a BD from the set of serving PBs can be given by
\begin{equation}
E_{r, S} = \sum_{m=0}^{S-1} \left[ \frac{P \beta_{0}}{\left( d^{2}\!+\!r^{2}\!-\!2 d r \cos\left( \theta_{o}\!+\!\frac{2 \pi m}{M} \right) \right)^{\frac{\delta}{2}}} X_{m} \right], \label{eq:energyReceived_multi}
\end{equation}
where $X_{m} \sim \Gamma (\kappa_{m}, \frac{1}{\kappa_{m}})$ is the fading random variable indexed according to the PB number. 

While semi-closed-form expressions \cite{Mos85} exist for the pdf and cdf of the sum of i.i.d. gamma random variables, a common approximation is to model the sum as another gamma random variable \cite{Cov14}, denoted here by $X_{\Sigma}$, where given $X_{i} \sim (\kappa_{i}, \tau_{i})$,
\begin{equation}
X_{\Sigma} = \sum_{i=1}^{n} X_{i} \sim \Gamma(\kappa_{\Sigma}, \tau_{\Sigma}), \label{eq:approx_gamma_A}
\end{equation}
with
\begin{equation}
\kappa_{\Sigma} = \frac{\mu^{2}}{\sum_{i=1}^{n} \kappa_{i} \tau_{i}^{2}}, \tau_{\Sigma} = \frac{\sum_{i=1}^{n} \kappa_{i} \tau_{i}^{2}}{\mu}, \textrm{and} \ \mu = \sum_{i=1}^{n} \kappa_{i} \tau_{i}.		\label{eq:approx_gamma_B}
\end{equation}
The combined shape parameter $\kappa_{\Sigma}$ may not be integer; hence, the outage probability is given as
\begin{equation}
P_{out} = \mathrm{Pr} (\gamma < \gamma_{th}) \approx F_{Z}\left( \frac{\gamma_{th}}{\beta} \right), \label{eq:Pout_multi}
\end{equation}
where $\beta$ is the equivalent path loss component of the SNR at the reader, and $Z$ is the fading random variable associated with $\gamma$, comprised of the product of random variables $X_{\Sigma}$ and $Y$. The cdf of $Z$ is given by
\begin{multline}
F_{Z}(z) = \frac{\pi \csc(\pi (\kappa_{X} - \kappa_{Y}))}{\Gamma(\kappa_{X}) \Gamma(\kappa_{Y})} \\ \times \left[ \Gamma(\kappa_{X}) (\kappa_{X} \kappa_{Y} z)^{\kappa_{X}} {}_{1}\tilde{F}_{2} (\kappa_{X}; \kappa_{X}\!+\!1, \kappa_{X}\!-\!\kappa_{Y}\!+\!1; \kappa_{X} \kappa_{Y} z) \right. \\
\left. + \Gamma(\kappa_{Y}) (\kappa_{X} \kappa_{Y} z)^{\kappa_{Y}} {}_{1}\tilde{F}_{2} (\kappa_{Y}; \kappa_{Y}\!+\!1, \kappa_{Y}\!-\!\kappa_{X}\!+\!1; \kappa_{X} \kappa_{Y} z) \right], \label{eq:productgamma_cdf}
\end{multline}
where $\kappa_{X}$ and $\kappa_{Y}$ are the shape parameters of $X_{\Sigma}$ and $Y$, respectively, and ${}_{p}\tilde{F}_{q}(\cdot; \cdot; \cdot)$ represents the regularized hypergeometric function. 

\subsection{Solution to the PB Placement Problem}

In this subsection, we outline methods to determine $d^{*}$ for the multiple serving PB case. Equation (\ref{eq:Pout_multi}) does not facilitate efficient solving of $d^{*}$. Therefore, we convert the multiple serving PB case to its equivalent geometric problem. The equivalent SNR threshold $\gamma_{eq,th}$ can then be defined similarly as (\ref{eq:conv_one}):
\begin{equation}
\gamma_{eq,th} = \frac{\gamma_{th}}{F_{Z}^{-1}(\varepsilon)},  \label{eq:conv_multi}
\end{equation}
where $F_{Z}^{-1}$ is the inverse of the cdf in (\ref{eq:productgamma_cdf}), and can be generated empirically. Similar to the single-PB case in (\ref{eq:conv_one}), this transformation is applicable to general Nakagami fading channels where the shape parameter may be non-integer. Note that $\gamma_{eq,th}$ here is an \textit{approximation} of the exact equivalent SNR threshold, due to the approximation used in summing the fading random variables in (\ref{eq:approx_gamma_A}).

We now present the corresponding result to Proposition 1 for $S \geq 2$ as follows. 

\begin{proposition}
For a BD deployed at distance $r$ from the reader, when $S \geq 2$, the minimum $\gamma_{eq}$ is obtained when the BD is located at the edge of the sector, i.e., $\theta_{o} = \frac{\pi}{M}$.
\end{proposition}

\begin{proof}
When each BD is served by multiple nearby PBs, the SNR expression is derived from the summation of the received signal components from all serving PBs, and can be given by:
\begin{align}
\gamma_{eq} &= \sum_{m=0}^{S-1} \frac{\alpha}{\left( d^{2} r^{2} + r^{4} - 2 d r^{3} \cos\left(\theta_{o} + \frac{2 \pi m}{M}\right) \right)^{\frac{\delta}{2}}}. \label{eq:approxSNR_multi}
\end{align}
The first derivative of (\ref{eq:approxSNR_multi}) is given by
\begin{equation}
\frac{\mathrm{d} \gamma_{eq}}{\mathrm{d} \theta_{o}} = \sum_{m=0}^{S-1} \frac{- \alpha d r^{3} \delta \sin\left( \theta_{o} + \frac{2 \pi m}{M} \right)}{r \left( d^{2} + r^{2} - 2 d r \cos\left(\theta_{o} + \frac{2 \pi m}{M}\right) \right)^{\frac{\delta}{2} + 1}}. \label{eq:prop2}
\end{equation}
Similar to Proposition 1, it can be shown that each term in the summation of (\ref{eq:prop2}) is an odd function that is positive for negative $\theta_{o}$ and vice versa, for any $\delta$. As a result, (\ref{eq:approxSNR_multi}) is symmetric about $\theta_{o} = 0$ and is monotonically increasing for negative $\theta_{o}$, and vice versa, with maximum occurring at $\theta_{o} = 0$. Therefore, minima must occur at $\theta_{o} = \pm \frac{\pi}{M}$.
\end{proof}

Next, a general algorithm for determining $d^{*}$ is first presented, followed by discussions on some specific values of $S$. 

\underline{\textit{General Algorithm}}: First, we perform linear search to find the value(s) of $r$ along the sector edge where (\ref{eq:approxSNR_multi}) equals $\gamma_{eq, th}$. This is equivalent to obtaining $r$ such that (\ref{eq:Pout_multi}) equals $\varepsilon$. While we aim to obtain an expression for $d$ in terms of $r$, there may be multiple values of $r$ where equality is attained. Thus, the correct solution of $r$ must be chosen, such that $\gamma_{eq}$ is greater than $\gamma_{eq, th}$ (or $P_{out}$ is greater than $\varepsilon$) up to that point. That is, the correct solution of $r$ is the smallest positive root of (\ref{eq:Pout_multi}). If an expression exists for the correct solution of $r$ for all $d$, it can be rearranged to be in the form $d(r)$, from which $d^{*}$ can be found via algebraic manipulations; otherwise numerical search is performed. The detailed steps are provided in Algorithm 1.

\begin{algorithm}
	\caption{Computation of the optimal PB placement distance $d^{*}$}
	\begin{algorithmic}[1]
		\REQUIRE Number of PBs, $M$; number of serving PBs, $S$; range of BD distances, $r \in [0, r_{upper}]$; SNR threshold, $\gamma_{th}$; outage probability threshold, $\varepsilon$.
		\STATE Perform linear search over the range of $r \in [0, r_{upper}]$ and compute (\ref{eq:approxSNR_multi}) for each $r$ to find intersection points of (\ref{eq:approxSNR_multi}) with $\gamma_{eq, th}$ in (\ref{eq:conv_multi}).
		\STATE For each $r$ in Step 1, pick the smallest positive intersection point and set the resulting expression or sequence of intersection points as $r_{cov}(d)$.
		\STATE Compute first derivative of $r_{cov}(d)$ numerically (or by algebraic manipulation, where an expression exists for $r_{cov}(d)$) and perform linear search over $d \in [0, \infty)$ to find the roots.
		\STATE Pick the smallest positive solution to Step 3 and set as $d^{*}$.
		\RETURN $d^{*}$
	\end{algorithmic}
\end{algorithm}

\underline{\textit{The special case of $S = 2$:}} Proposition 2 applies for any integer $S$; that is, a BD located on the sector edge at distance $r$ from the reader achieves the lowest equivalent SNR compared to any other location with distance $r$ from the reader, regardless of the number of serving PBs. Applying this observation to the case where a BD on the sector edge is served by two nearest PBs, the achievable SNR when $S = 2$ is twice that of $S = 1$. Hence, when $S = 2$, 
\begin{equation}
\gamma_{eq} = \frac{2 \alpha}{r^{\delta} \left( d^{2} + r^{2} - 2 d r \cos\left(\frac{\pi}{M}\right) \right)^{\frac{\delta}{2}}}.	\label{eq:eqSNR_two}
\end{equation}
Therefore, we can directly invoke Theorem 1 to compute $d^{*}$ by setting $\varsigma = \left(\!\frac{2 \alpha}{\gamma_{eq,th}}\!\right)^{\frac{2}{\delta}}$, where $\gamma_{eq,th}$ refers to the equivalent SNR threshold for $S = 2$. Note that $\gamma_{eq,th}$ is constant for any values of $r$ and $d$ only when $S = 1$ or $2$, as $F_{XY}^{-1}$ in (\ref{eq:conv_one}) does not depend on the path losses between individual PBs and the BD. Similarly, Corollary 1 also holds. 

\underline{\textit{The special case of $S = M$:}} For $S > 2$, we observe that $\gamma_{eq}$ in (\ref{eq:conv_multi}) needs to be computed for every combination of $r$ and $d$, as random variable $X_{\Sigma}$ accounts for the path losses between each PB and the BD, and thus closed-form expressions for $d^{*}$ and $r_{cov}$ are unlikely to exist. Here we focus on the case of $S = M$, as it theoretically represents the best possible coverage performance. We can exploit symmetry about the sector edge, on which the BD is located, to deduce the total energy contribution of the PBs beyond the first and second nearest PBs. For example, the third and fourth nearest PBs have angular distance $\frac{3\pi}{M}$ from the BD. The fifth and sixth PBs have angular distance $\frac{5\pi}{M}$ from the BD, and so on. Therefore, the general equivalent SNR expressions for even and odd $M$ are 
\begin{align}
\gamma_{eq} &= 2 \alpha \sum_{m=0}^{M/2-1} r^{-\delta} \left( d^{2} + r^{2} - 2 d r \cos\left(\frac{(m+1) \pi}{M}\right) \right)^{-\frac{\delta}{2}}, 	\label{eq:evenS} \\
\gamma_{eq} &= \frac{\alpha}{r^{\delta} \left( d^{2} + r^{2} + 2 d r \right)^{\frac{\delta}{2}}} \nonumber \\
		& \quad + 2 \alpha\!\sum_{m=0}^{\left\lfloor{M}/2\right\rfloor-1} r^{-\delta}\!\left(\!d^{2}\!+\!r^{2}\!-\!2 d r \cos\left(\frac{(m+1) \pi}{M}\right)\! \right)^{-\frac{\delta}{2}}. 	\label{eq:oddS}
\end{align}
Noting the symmetry of the PBs' individual contributions, $d^{*}$ may be approximated by setting $\varsigma = \left(\!\frac{\chi \alpha}{\gamma_{eq,th}}\!\right)^{\frac{2}{\delta}}$, where $\chi$ is the ratio between $\gamma_{eq}$ under $S = M$ to $\gamma_{eq}$ under $S = 1$.


\section{Numerical Results}

In this section, we numerically evaluate the impacts of the number of PBs, the PB distance and channel parameters on the GCD of the bistatic backscatter network. Cases where each BD is served by one and multiple nearest PBs are considered. Unless otherwise specified, the carrier frequencies and transmit powers of all PBs are $915$ MHz and $27$ dBm, respectively, and BDs are modeled as semi-passive devices. We adopt $\delta = 2.4$ as the baseline path loss exponent, which is a typical value for natural environments such as various types of forests and grassland where the nodes are located between $1$-$2$ m above ground \cite{GF13}. Where fading is considered, Nakagami fading with parameter $4$ is assumed; the SNR threshold for all tags is $\gamma_{th} = 5$ dB and the QoS constraint is $\varepsilon = 0.05$; and the receive noise is $\sigma^{2} = -110$ dBm. The tag reflection efficiency is $\eta = 0.49$ and the antenna gain is $2.1$ dBi; the structural mode is set to $A = 0.6047 + j0.5042$ for all BDs \cite{Kim14}; and the reflection coefficients are set to $\{G_{0}, G_{1}\} = \{A, -\frac{A}{|A|}\}$. BDs are attached to electrically non-conductive materials, resulting in negligible gain penalty; and the polarization mismatch loss is $0.8$ for both the forward and backscatter links, to account for imperfections in BD orientations. Each symbol transmitted by a BD occurs over $L = 20$ source samples.

\subsection{Optimal PB Distance and GCD Characteristics}

\begin{figure}[h!]
	\centering
	\begin{subfigure}{0.495\textwidth}
		\centering
		\includegraphics[width=3.5in]{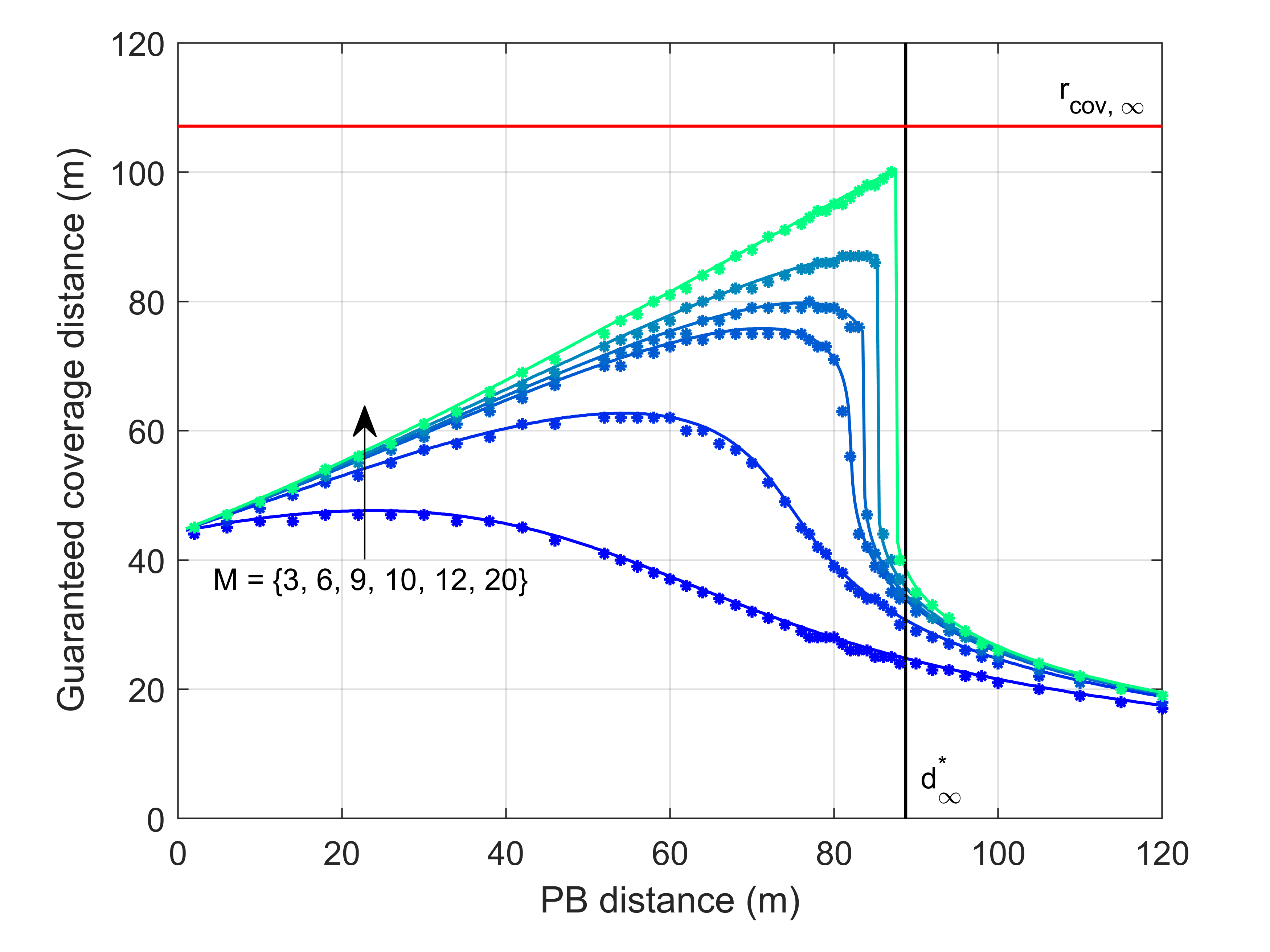}
		\caption{Algorithm 1 vs. Simulation}
	\end{subfigure} \hfill
	\begin{subfigure}{0.495\textwidth}
		\centering
		\includegraphics[width=3.5in]{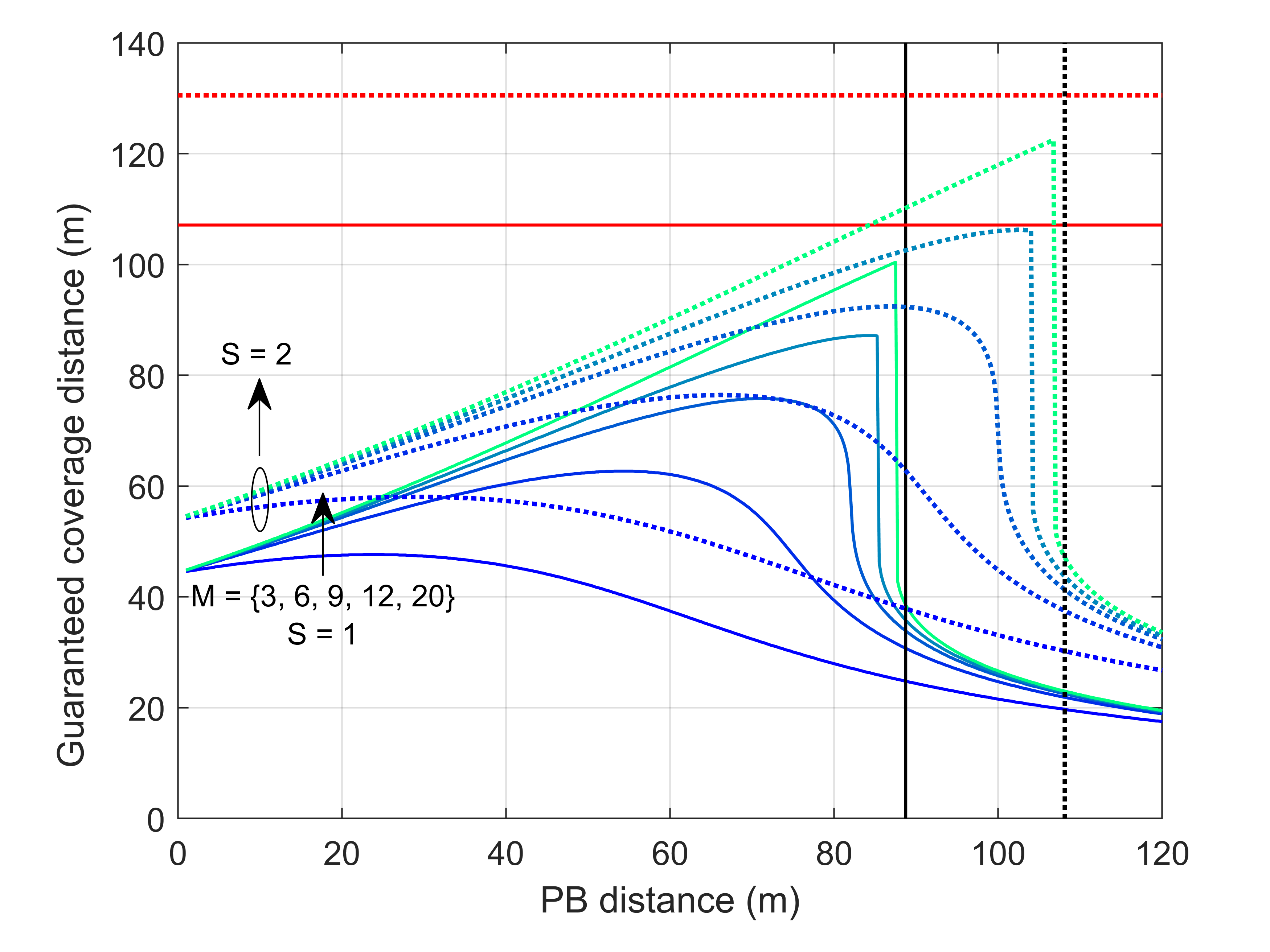}
		\caption{$S = 1$ vs. $S = 2$}
	\end{subfigure}
	\caption{Behavior of the GCD for varying $M$, and at various PB distances.}
	\label{fig:singlePBEnsemble}
\end{figure}

Fig. \ref{fig:singlePBEnsemble}(a) shows the GCD as a function of $d$, where $S = 1$, obtained via both simulation (asterisks) and Algorithm 1 (more specifically, the first two steps of Algorithm 1, which obtains the GCD for a given $d$; solid lines). The simulation results are obtained using $2 \times 10^{4}$ realizations of BD locations and channel coefficients. It is evident that the results obtained using Algorithm 1 match the simulation results. The existence of an optimal PB distance is confirmed, where the optimal $d$ that maximizes the GCD in each curve matches with $d^{*}$ obtained using Theorem 1. The GCD increases with $d$ up to a maximum, before experiencing a sharp drop. The sharpness of the drop in the GCD beyond $d^{*}$ increases with the total number of PBs. Therefore, in practice, it is appropriate to deploy PBs at or at slightly less than the optimal distance, without much reduction in the GCD, but it is unwise to deploy PBs beyond the optimal distance. The upper bounds to $d^{*}$ and the GCD, presented in Corollary 1, are shown using the vertical (black) and horizontal (red) lines, respectively, and it is clear that both quantities converge to their upper bounds as $M$ increases. Hereafter, we present results obtained using Algorithm 1 only.

The characteristics of the GCD for the case of $S = 2$ is shown in Fig. \ref{fig:singlePBEnsemble}(b). One can observe the improvement in the GCD arising from one additional serving PB where, for example, under the case of $M = 6$ with $d = 50$ m, the GCD with a single serving PB is around $63$ m; whereas the GCD with two serving PBs increases to $76$ m, i.e., around $20\%$ improvement. We can also see that the asymptotic GCD increases from $108$ m to $130$ m, i.e., a $20\%$ improvement.

\begin{figure}[h!]
	\centering
	\begin{subfigure}{0.495\textwidth}
		\centering
		\includegraphics[width=3.5in]{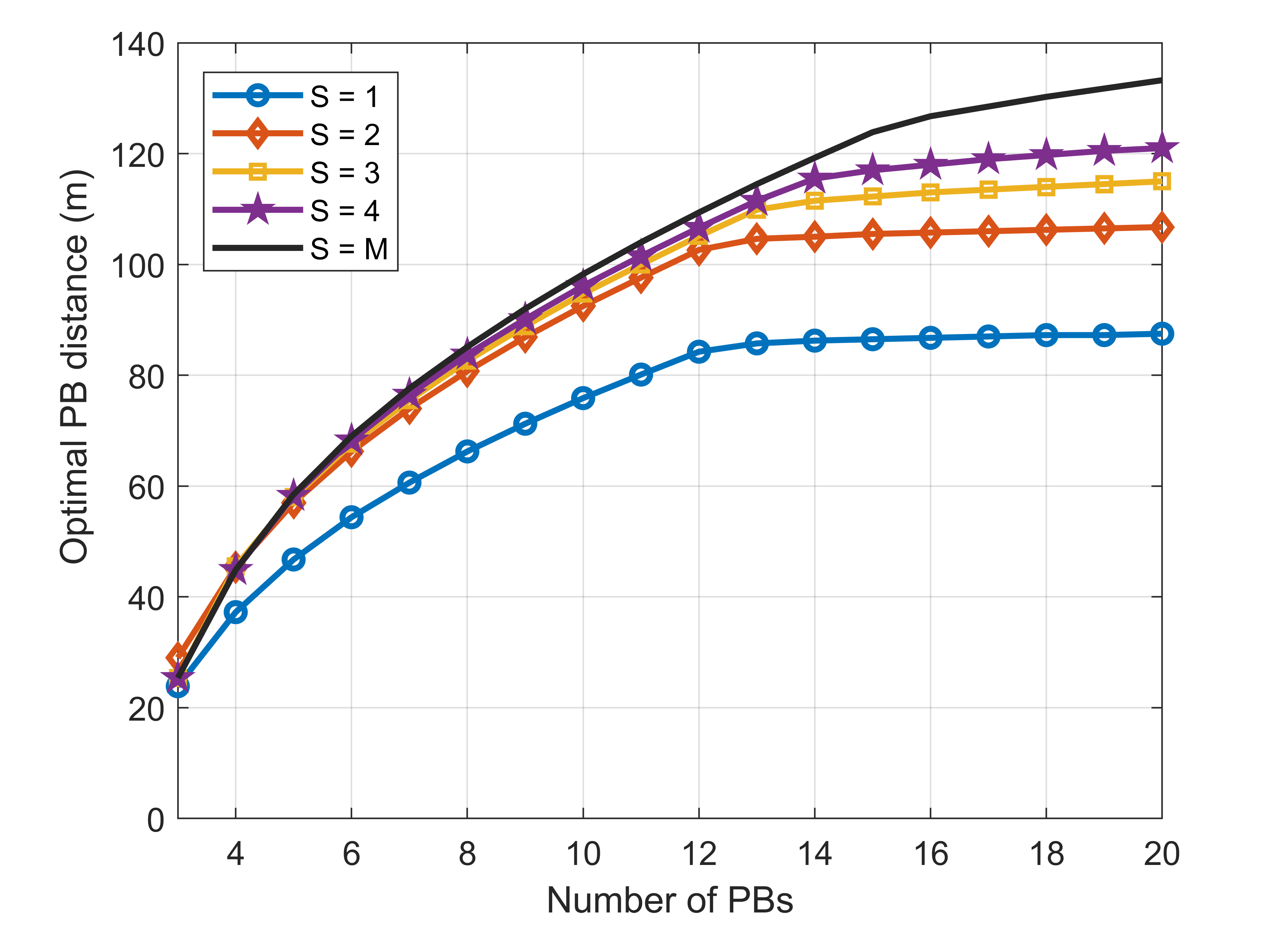}
		\caption{Optimal PB distance}
	\end{subfigure} \hfill
	\begin{subfigure}{0.495\textwidth}
		\centering
		\includegraphics[width=3.5in]{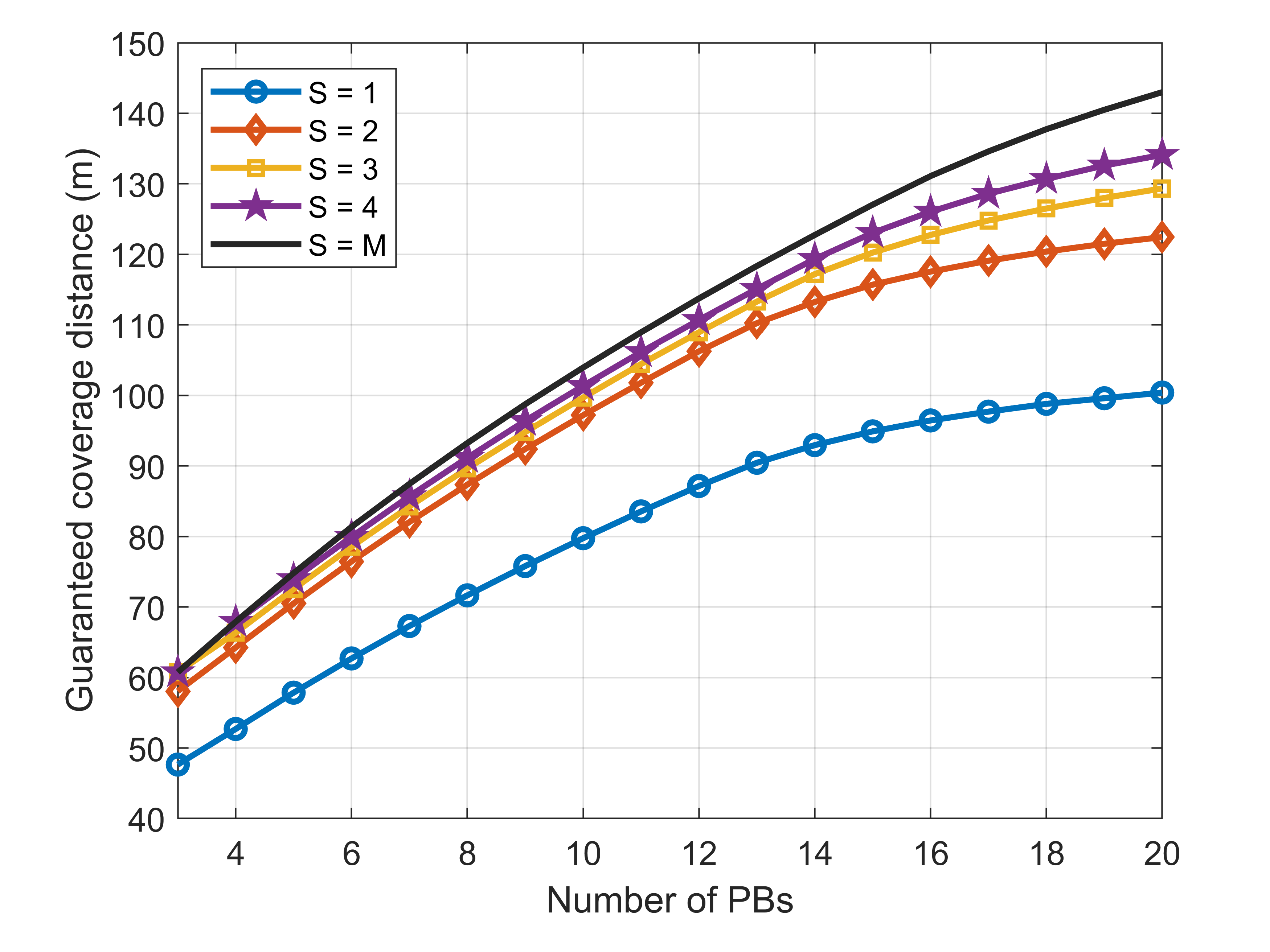}
		\caption{Guaranteed coverage distance}
	\end{subfigure}
	\caption{Effect of the number of serving PBs on the optimal PB distance and GCD.}
	\label{fig:optimalDistances}
\end{figure}

Fig. \ref{fig:optimalDistances} shows the optimal PB distance and the corresponding maximum GCD associated with the number of PBs. The benefit of increasing the number of serving PBs on the GCD diminishes as more PBs are required to serve each BD, as shown by the smaller increases in both $d^{*}$ and GCD when $S$ increases beyond $2$. For a network with a small number of PBs ($M \leq 12$), our results suggest to use no more than $2$ serving PBs for each BD.

\subsection{Effect of Circuit Power Constraint}

So far, our analysis and results have assumed semi-passive BDs with built-in batteries to support the circuit operation. In this subsection, we consider the case of passive BDs which require a minimum amount of energy from incoming signals in order to operate their circuits.

The addition of a circuit power constraint results in a change in the tag architecture. An energy harvesting module is required to harvest a portion of the energy in the incident signal to sustain the circuit operation, with the remaining part of the signal used for backscattering. Moreover, the efficiency of the harvester is often a nonlinear function of the incident signal energy, resulting in further losses and changes to the SNR expression in Section III. For the single serving PB case, the modified SNR expression becomes
\begin{multline}
\gamma = \Bigg\{ \zeta \Bigg( \frac{P \beta_{0} G_{T} \chi_{CT}}{\left(d^{2} + r^{2} - 2 d r \cos \left( \frac{\pi}{M} \right)\right)^{\frac{\delta}{2}}} X \Bigg) - \xi \Bigg\}^{+} \\ \times \left( \frac{\eta \beta_{0} G_{T} \chi_{TR} |b_{0} - b_{1}|^{2} L}{r^{\delta} N_{0}} \right) Y,	\label{eq:SNR_one_xi}
\end{multline}
where $G_{T}$ and $\chi_{CT}, \chi_{TR}$ are the antenna gains at the BD, and the polarization mismatch losses between PB-BD and BD-reader, respectively; $\zeta(\cdot)$ denotes the nonlinear energy harvesting efficiency function; and $\{ x \}^{+} = \max\{0, x\}$. The difference between (\ref{eq:SNR_one_xi}) and $\gamma$ in (\ref{eq:eqSNR_one}) is the inclusion of the circuit's minimum power consumption $\xi$, which is subtracted from the remaining power in the PB-BD link after harvesting. As such, the expression in the first bracket indicates the energy available for backscattering after the circuit's power consumption has been fulfilled, and is zero if the received energy is insufficient. The addition of the circuit power constraint significantly impacts the mathematical tractability of solving for $d^{*}$, where the distribution of the energy available for backscattering admits a complex form that depends on the harvesting efficiency function. We adopt an approximation of the harvesting efficiency function of the rectenna model presented in \cite{R1}, with the efficiency function given by
\begin{equation}
\zeta(p) = \begin{dcases}
		\frac{0.4911}{1 + \exp(-0.2(p + 21.65))}, & p \geq -42, \\
		0, & \text{otherwise},
		\end{dcases}
\end{equation}
where $p$ is the input power to the harvester in dBm, and $\zeta(p)$ also has units of dBm. For the subsequent result, we use a carrier frequency of $868$ MHz to match the rectenna model, although we note that these results are readily generalizable to adjacent carrier frequencies with appropriate energy harvester design.

\begin{figure}[h!]
	\begin{subfigure}{0.495\textwidth}
		\includegraphics[width=3.5in]{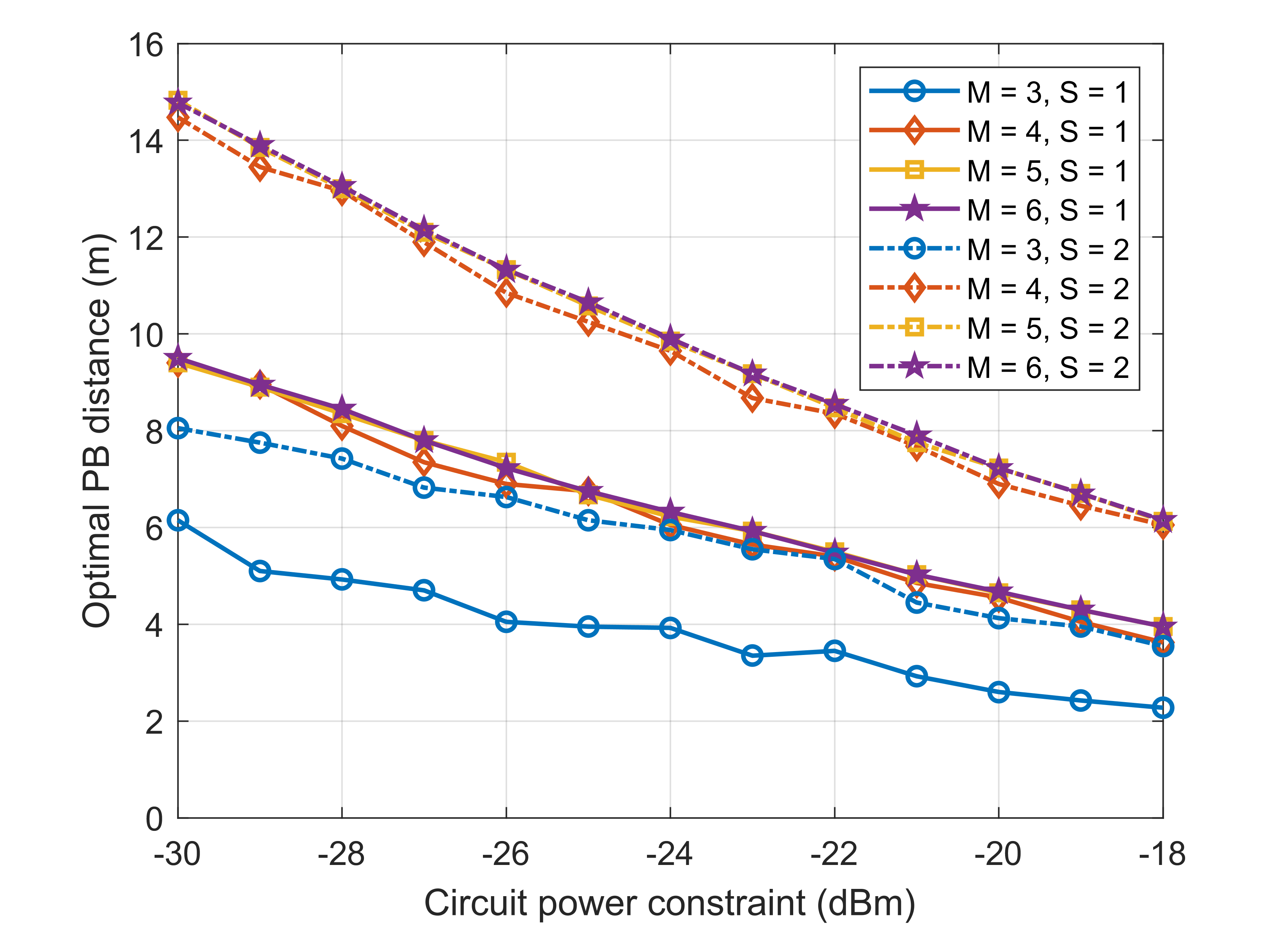}
		\caption{Optimal PB distance}
	\end{subfigure} \hfill
	\begin{subfigure}{0.495\textwidth}
		\includegraphics[width=3.5in]{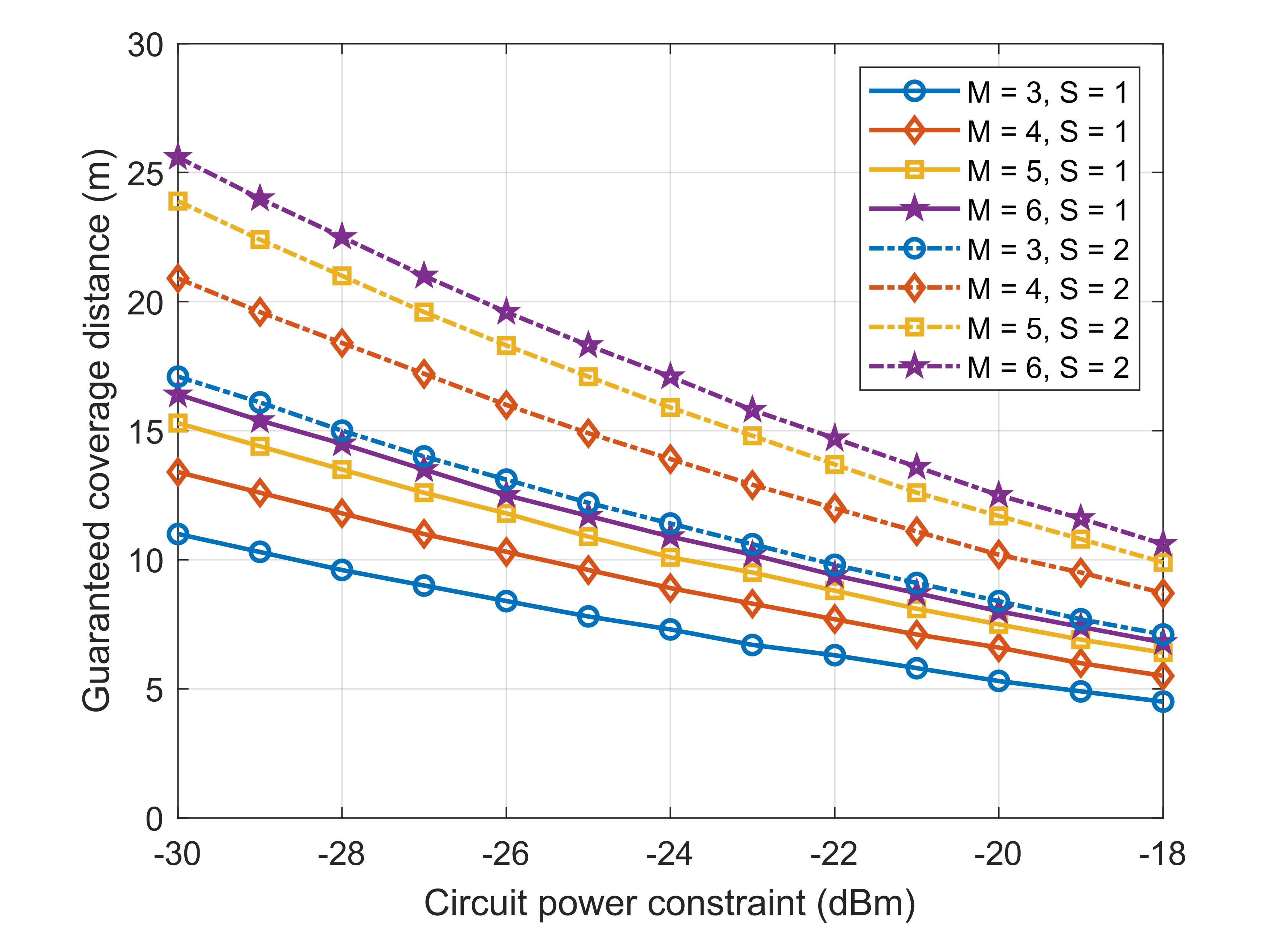}
		\caption{Guaranteed coverage distance}
	\end{subfigure}
	\caption{Effect of BD circuit power constraint the optimal PB distance and GCD.}
	\label{fig:cps}
\end{figure}

Fig. \ref{fig:cps} highlights the effects of the BD circuit power constraint on the optimal PB distance, with the results obtained numerically using Algorithm 1. For the purpose of illustration, the transmit power of PBs is set to $35$ dBm. Evidently, an increase in the circuit power constraint results in a smaller optimal PB distance and GCD, where a one order-of-magnitude increase from $-30$ dBm to $-20$ dBm results in between half and two-thirds reduction in the GCD depending on the serving PB configuration. In addition, the optimal PB distance does not increase further when $M > 5$, as shown in Fig. \ref{fig:cps}(a), where the curve for $M = 5$ completely overlaps with the curve for $M = 6$. Nevertheless, adding more PBs beyond $M > 5$ still results in a small increase in the GCD, as shown in Fig. \ref{fig:cps}(b), but the small increase may not justify the cost of deploying more PBs when $M$ is already large. On the other hand, having $S = 2$ presents a significantly improved GCD over the case of $S = 1$, similar to the observations in earlier numerical results.

\subsection{Comparisons with Alternative Placement Schemes}

\begin{figure}[h!]
	\centering
	\includegraphics[width=3.5in]{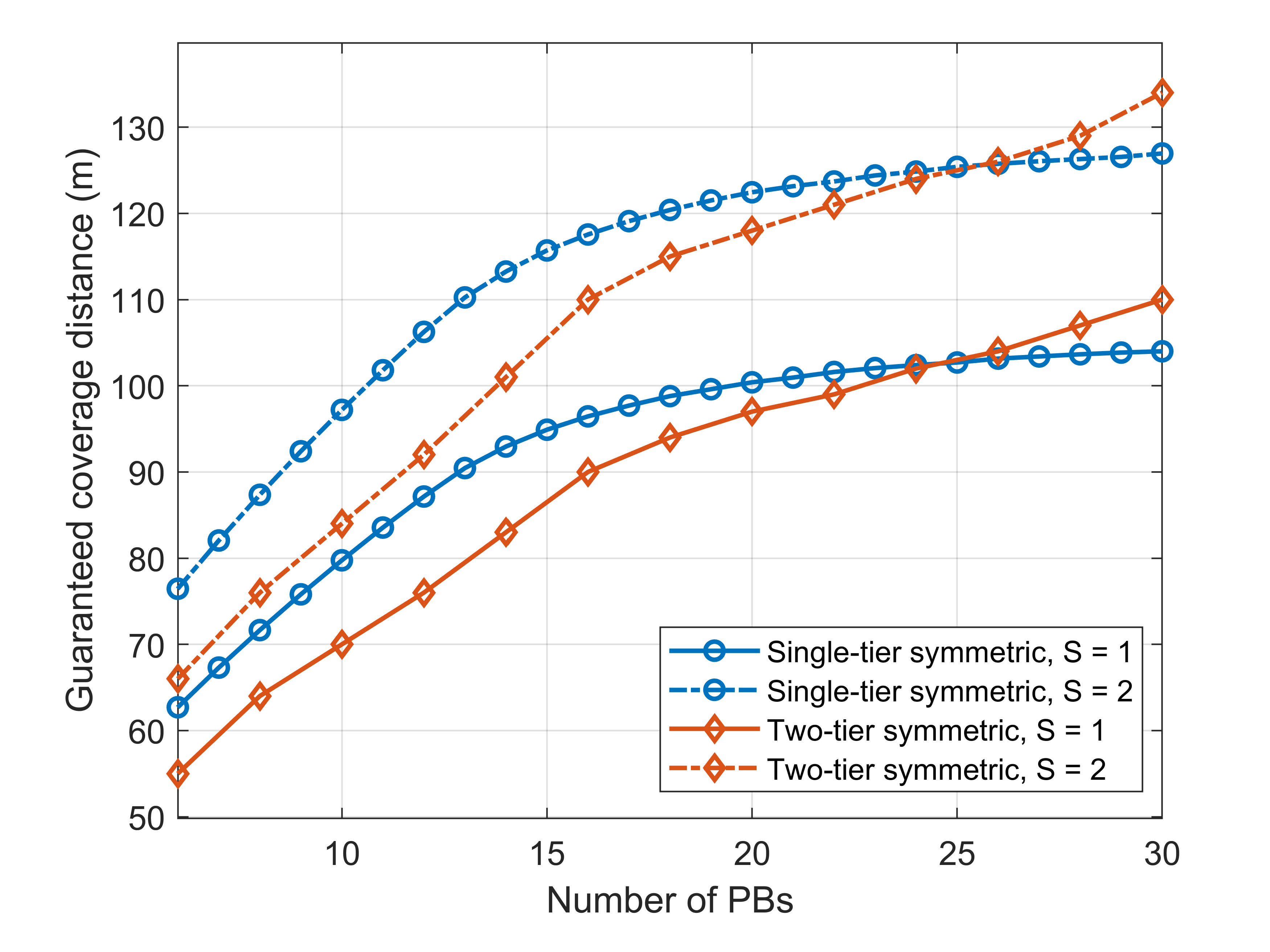}
	\caption{Comparison of the GCD between the single-tier symmetric and two-tier symmetric placement schemes for $S = 1$ and $S = 2$.}
	\label{fig:opt_comp}
\end{figure}

The optimality of the single-tier PB placement scheme adopted in this work may be intuitive to understand for a small number of PBs. Nevertheless, one may expect that a two-tier symmetric scheme should outperform the single-tier scheme when $M$ is large. Therefore, it is important to compare the performance between the single-tier and two-tier schemes. Fig. \ref{fig:opt_comp} compares the maximum achievable GCD for the single-tier symmetric scheme studied in this paper against a two-tier symmetric scheme. Under the two-tier scheme, each tier has at least $3$ PBs. Hence, the number of PBs in the inner and outer tiers are $M_{1} = \{3, \ldots, M-3\}$ and $M_{2} = M - M_{1}$, respectively, where $M \geq 6$. The distance of the inner tier is set to $d^{*}$ from the single-tier scheme (from Theorem 1) when $M_{1}$ PBs are deployed; while the distance of the outer tier is numerically optimized over a range within the vicinity of this value. At each outer tier distance, the rotation offset of the second-tier PBs relative to the positive $x$-axis is varied between $[0, \frac{\pi}{M_{2}}]$. For each value of $M$, the largest GCD achieved by any combination of \{$M_{1}$, $M_{2}$, outer tier distance, rotation offset\} is numerically obtained and shown in Fig. \ref{fig:opt_comp}. One can observe that for both $S = 1$ and $S = 2$, single-tier outperforms two-tier up to at least $M = 24$, which is already a relatively large number of PBs. Beyond this threshold, the two-tier scheme has a performance advantage in cases where the optimal placement allocates more PBs to the outer tier than the inner tier. However, the computation of an optimal two-tier placement incurs much higher complexity than Algorithm 1 due to the additional variables involved. Moreover, it is usually not wise in practice to deploy a very large number of PBs in multiple tiers surrounding a single reader in order to grow the coverage area, as the gains in coverage area are limited with more added PBs. When wide-area coverage is concerned, a more efficient network deployment strategy (discussed in Section V-D) is the use of multiple readers, each of which is surrounded by PBs placed according to the single-tier scheme. Thus, the single-tier scheme balances good performance with reasonable deployment costs.

\begin{figure}[h!]
	\centering
	\includegraphics[width=3.5in]{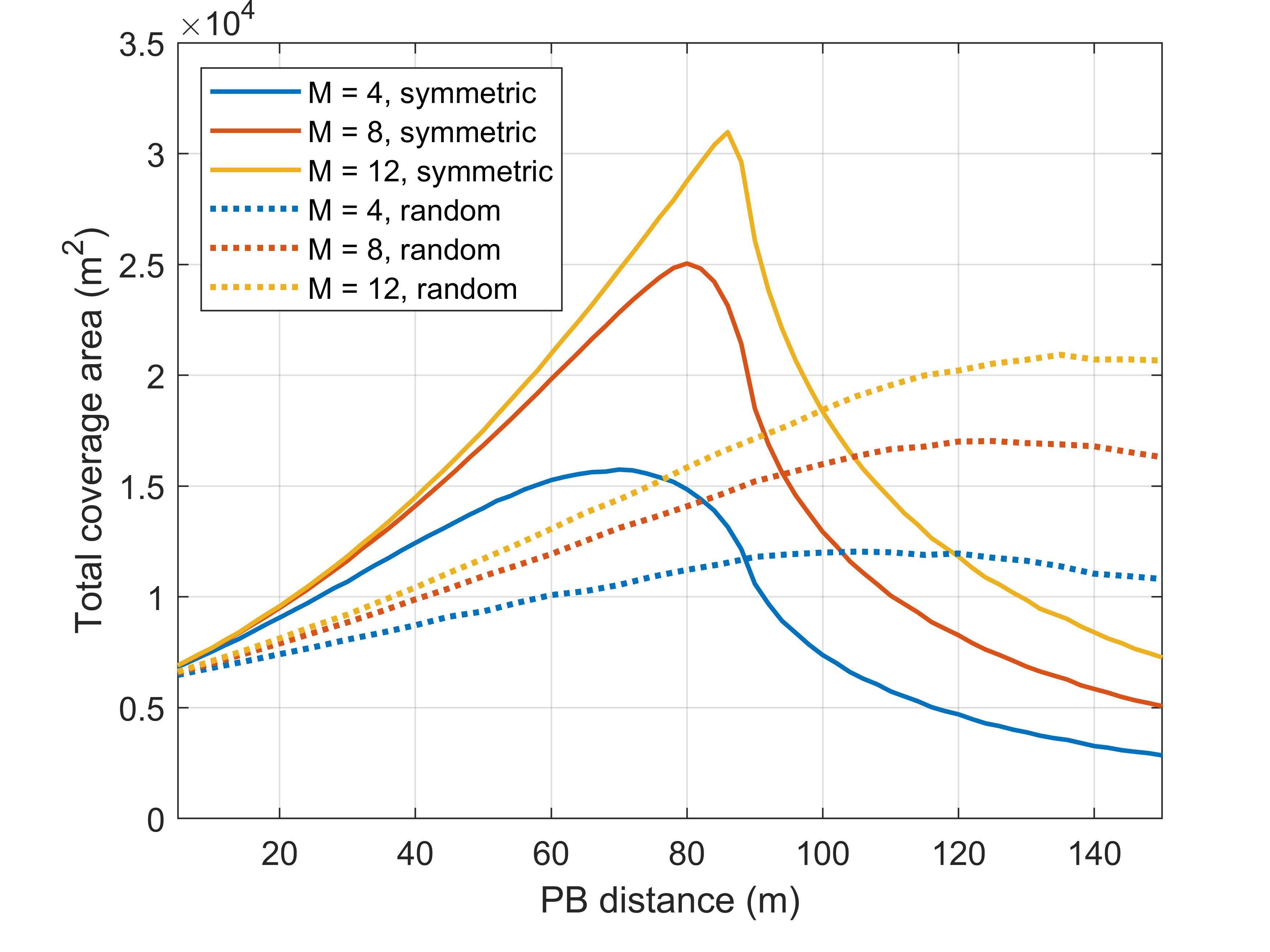}
	\caption{Comparison of the total coverage area between the symmetric and random PB placement schemes, where each PB is served by its nearest PB (i.e., $S = 1$).}
	\label{fig:cov_rand}
\end{figure}

We also examine the performance comparison with a random placement scheme, which is typically considered in the literature as a useful benchmark. Note that the GCD is no longer a suitable metric for random placement, which may result in scattered `uncovered' regions closer to the reader while regions further away from the reader are covered. For fair comparison, we use the total coverage area as the metric. Note that for the symmetric placement scheme, the total coverage area includes the additional `covered' regions beyond the GCD. Fig. \ref{fig:cov_rand} presents the numerically computed total coverage areas for both the symmetric and random placement schemes. The PB distance $d$ for the symmetric placement scheme is the same as defined previously in this paper. The PB distance $d$ for the random placement scheme denotes the maximum distance for PB placement, where $M$ PBs are randomly placed according to a uniform distribution in a disk region of radius $d$ centered at the reader. Hence, the $x$-axis in Fig. \ref{fig:cov_rand} refers to different quantities for the two schemes. It is clear that when all areas satisfying the SNR threshold are accounted for, the symmetric placement scheme outperforms the random scheme in terms of the maximum achievable total coverage area (i.e., the total coverage area achieved at the optimal distance parameter value), by $25$-$35\%$ for the values of $M$ considered.

\subsection{Extension to Multiple Readers and General Design Guidelines}

The single-reader nature of the proposed bistatic backscatter system and the associated PB distance optimization ensures coverage for a reasonably-sized area on the order of tens of thousands of $m^{2}$. These results are readily extendable when coverage over larger areas is required. To achieve this, multiple instances of the single-tier symmetric PB setup may be deployed in a cellular fashion, wherein each cell contains its own reader and PBs deployed at the optimal $d^{*}$ relative to its reader. As an example, Fig.~\ref{fig:multicell} shows that the addition of multiple cells surrounding a single initial cell results in a considerably larger overall coverage area.

\begin{figure}[h!]
	\centering
	\includegraphics[width=3.5in]{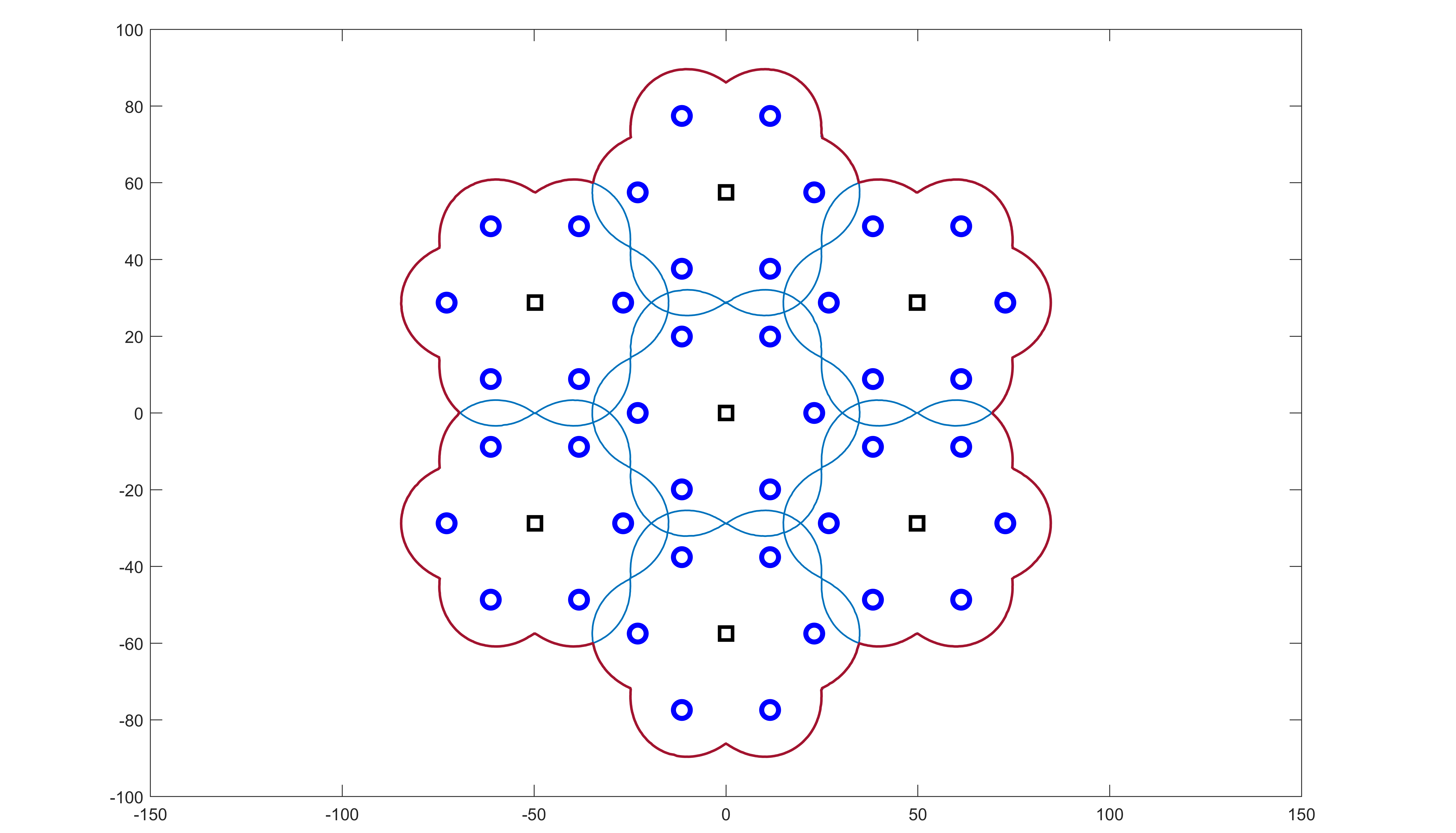}
	\caption{Illustration of a multi-cell bistatic backscatter network, with each cell consisting of a reader with PBs deployed in the single-tier symmetric scheme. The coverage area of each cell is shown in blue and the overall coverage area is shown in red.}
	\label{fig:multicell}
\end{figure}

As the single-tier PB analysis accounts for channel variations to provide a deterministic GCD for each cell, the parameters $d^{*}$ and $M$ may be chosen for any one cell based on the analysis therein, and applied to all other cells. Further optimization may be performed by taking into account that in certain cases, a PB placed near the boundary of two cells can be utilized to serve both cells, further reducing the overall number of PBs required in the whole network. Such additional optimizations deserve further study, but are beyond the scope of this work.

Thus, we find that our design insights in Sections III-IV provide the general guideline to deploy the PBs in a ring with equal separation to ensure maximization of guaranteed coverage for a network with a single reader. Where extended coverage is required, the same design parameters (i.e., $d^{*}$ and $M$) can also be obtained using the same procedure for any number of serving PBs, with the optimized parameters applied to each cell in a multi-cell network. With appropriate inter-cell interference mitigation strategies, different cells in the system are able to work jointly to provide a larger overall coverage region.


\section{Conclusion}
In this paper, a single-tier symmetric PB placement strategy for a bistatic backscatter network was studied, with the objective of maximizing the guaranteed coverage distance. Expressions for the outage probability were derived for the cases where a BD is served by one or multiple nearest PBs. In addition, closed-form expressions for the optimal PB placement distance were shown to be attainable for the cases of serving with one and two nearest PBs. Notably, when the number of PBs in the network is small, the use of two nearest PBs as serving PBs is a favorable choice. Moreover, the advantages of the single-tier symmetric placement scheme over other schemes, such as the two-tier symmetric and random schemes, were also demonstrated. This paper presents an introductory study into the efficient infrastructure deployment in bistatic backscatter networks. Future work can extend towards optimal placement of both PBs and readers to cover larger geographical areas.


\appendices
\section{Proof of Theorem 1}
We wish to solve the equation $\gamma_{eq} = \gamma_{eq,th}$, and maximize the resulting expression to obtain $d^{*}$. It is already established in Proposition 1 that for a BD at distance $r$ from the reader, the location where it achieves the lowest SNR must lie on the the sector edge; that is, $\theta_{o} = \frac{\pi}{M}$. The equation to be solved becomes a quartic polynomial by rearranging $\gamma_{eq}$ in (\ref{eq:eqSNR_one}):
\begin{equation}
f(r) = r^{4} - 2 d r ^{3}\cos\left(\frac{\pi}{M}\right) + d^{2} r^{2} - \varsigma, \label{eq:f(r)}
\end{equation}
where we have defined $\varsigma = \left( \frac{\alpha}{\gamma_{eq,th}} \right)^{\frac{2}{\delta}}$. Equation (\ref{eq:f(r)}) is analogous to a path loss function, and has the same roots as the equation $\gamma = \gamma_{th}$. Hereafter, we denote $\cos\left(\frac{\pi}{M}\right)$ as $\Theta$ for compactness. We first present the following result, which aids the remainder of the proof.

\begin{lemma}
The equivalent SNR of a BD on the sector edge is monotonically decreasing as $r$ increases, for $M \leq 9$. The monotonicity condition does not hold for $M \geq 10$.
\end{lemma}

\begin{proof}
We consider the domain $r \in [0, \infty)$. The roots of the first derivative of (\ref{eq:f(r)}) with respect to $r$ tell us the segments on the sector edge where the SNR is increasing or decreasing:
\begin{equation}
\frac{\mathrm{d}\varsigma}{\mathrm{d}r} = 4 r^{3} - 6 d r^{2} \Theta + 2 d^{2} r. \label{eq:f'(r)}
\end{equation}
Equation (\ref{eq:f'(r)}) tends to $\infty$ as $r$ increases, and can be simplified to a (convex) quadratic $4 r^{2} - 6 d r \Theta + 2 d^{2}$ by observing that there is a zero root. The discriminant of this quadratic, $9 d^2 \Theta^{2} - 8 d^2$, is less than zero for all $d$ when $\Theta < \frac{7}{9}$; taking the floor of the resulting inequality solution gives $M \leq 9$. When this is true, (\ref{eq:f'(r)}) has two complex roots in addition to the zero root, and is hence positive for all $r > 0$. It follows that the discriminant must be positive for some $d$, when $M \geq 10$. As the path loss function in (\ref{eq:f(r)}) is an inverted form of the SNR, the SNR is non-increasing for $M \leq 9$, and non-monotonic otherwise.
\end{proof}

Deferring the case where both non-zero roots of (\ref{eq:f'(r)}) are real and equal to the sequel, we can observe from Lemma 1 that there are two cases regarding the two roots in question:
\begin{itemize}
\item If both real, then along the sector edge, $\gamma_{eq} > \gamma_{eq, th}$ for small $r$, then $\gamma_{eq} < \gamma_{eq, th}$ for some time where $r \approx d$, and then $\gamma_{eq} > \gamma_{eq, th}$ as $r \rightarrow \infty$. In this case, (\ref{eq:f(r)}) has three positive real roots, where $f(r) > 0$ between the first and second roots (corresponding to a coverage gap); and $r_{cov}$ is equal to the smallest positive real root.
\item If both complex, then the SNR is monotonically decreasing for all $r$; and $r_{cov}$ is the single positive real root of (\ref{eq:f(r)}).
\end{itemize}
Let $r_{max}$ be the largest positive real root to (\ref{eq:f(r)}). For $M \leq 9$, where only one positive real root exists for $f(r)$, $r_{cov}$ is equal to $r_{max}$. For $M \geq 10$, depending on $M$ and $\gamma_{eq,th}$, $r_{cov}$ may correspond to a smaller root, since continuous coverage on the sector edge is not guaranteed.

\begin{figure}[h!]
	\centering
	\includegraphics[width=3.5in]{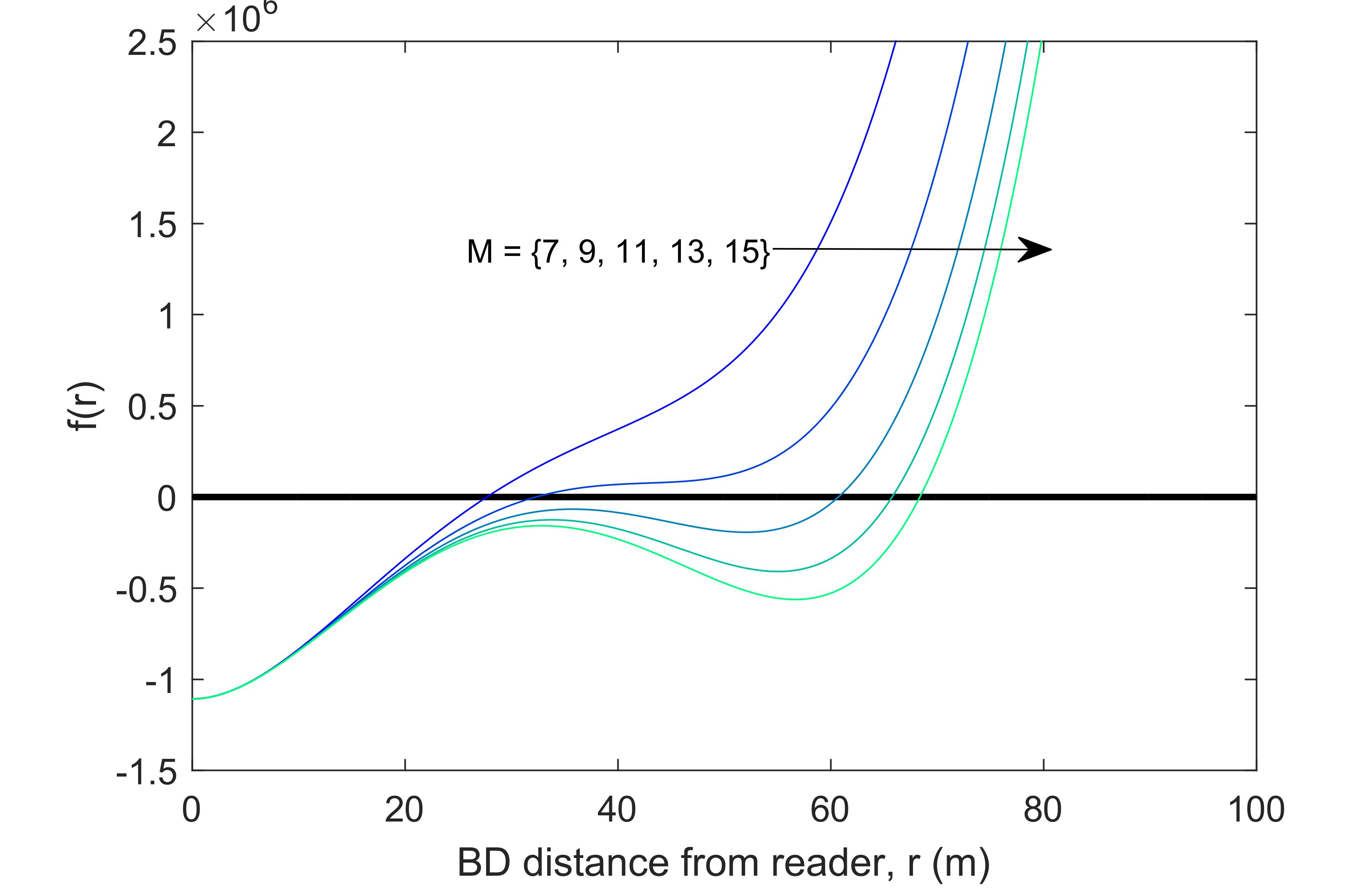}
	\caption{Behavior of path loss function in (\ref{eq:f(r)}) as $M$ increases.}
	\label{fig:app1}
\end{figure}

The behavior of $f(r)$ is shown in Fig. \ref{fig:app1} for representative values of $d$, $\varsigma$ and $M$. Note that the remaining root is located in the negative $x$-axis. The occurrence of three positive roots depends on the value of $\gamma_{th}$, which translates to a vertical shift in the path loss function in (\ref{eq:f(r)}). It is evident that (\ref{eq:f(r)}) is non-decreasing for $M \leq 9$, and that the coverage gap behavior, occurring when three positive real roots exist, becomes apparent for $M \geq 11$ and certain values of $\varsigma$. 

\begin{figure}[h!]
	\centering
	\begin{subfigure}{0.48\textwidth}
		\centering
		\includegraphics[width=3.5in]{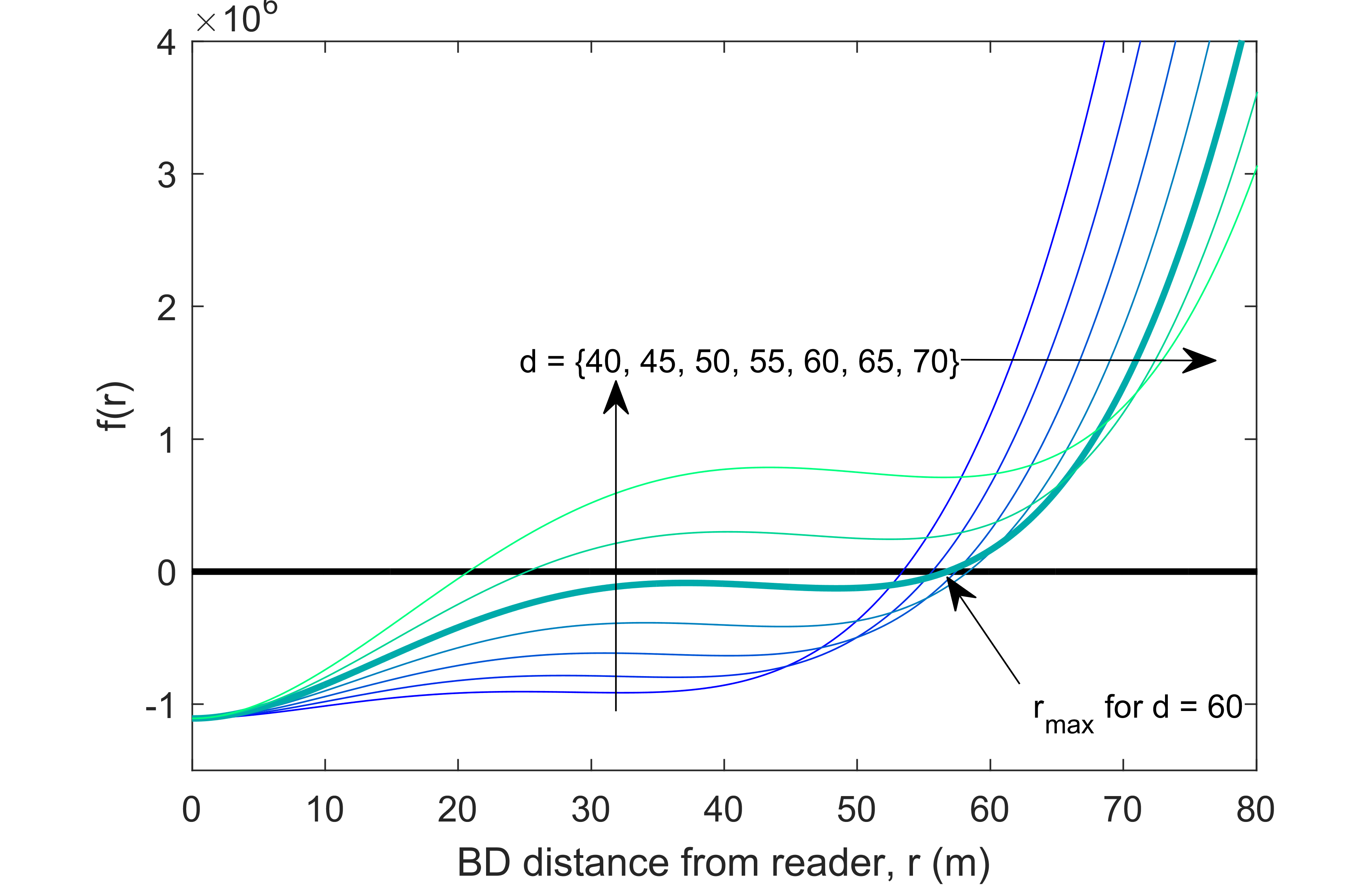}
		\caption{$M = 10$}
	\end{subfigure} \hfill
	\begin{subfigure}{0.48\textwidth}
		\centering
		\includegraphics[width=3.5in]{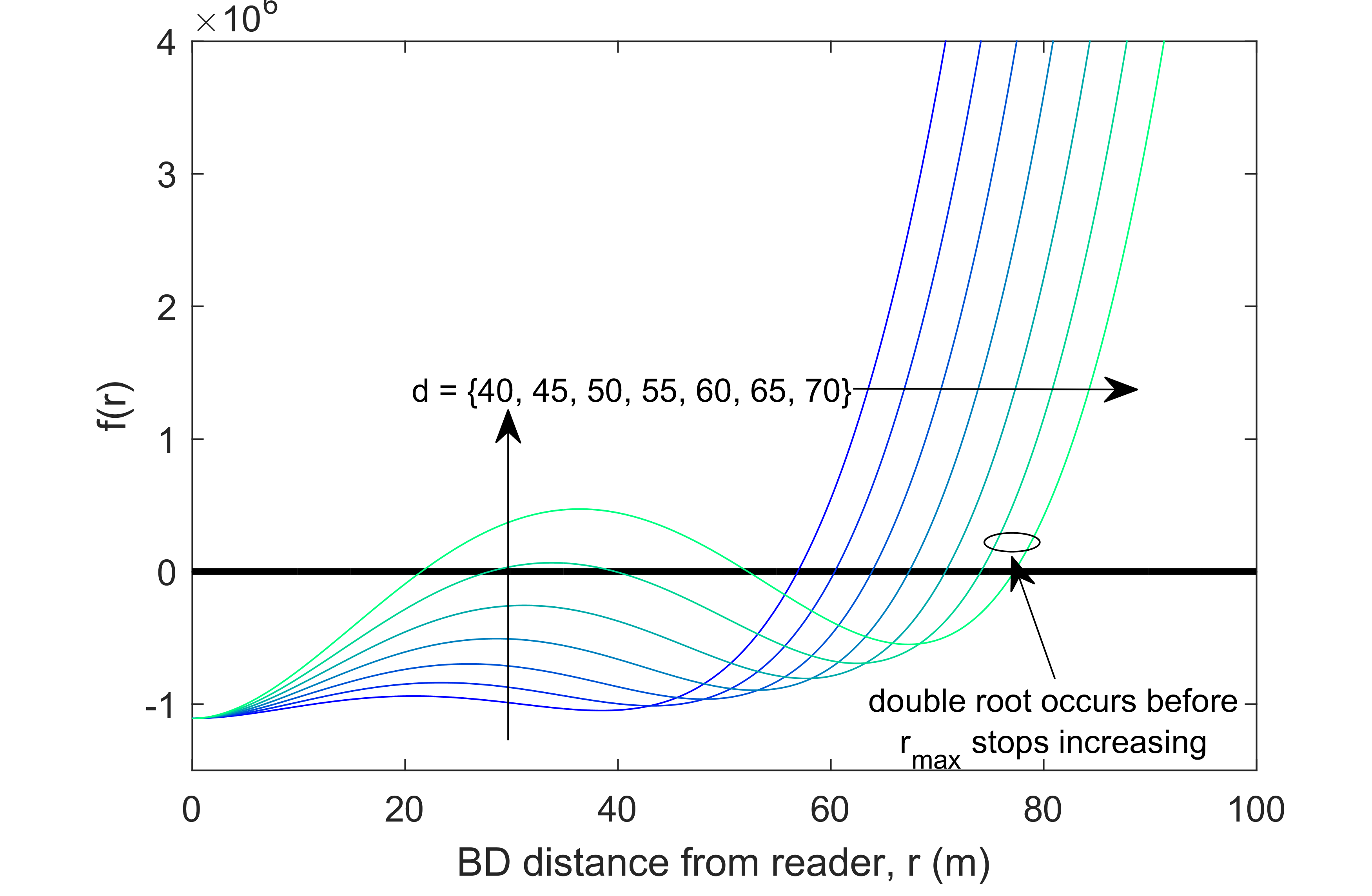}
		\caption{$M = 20$}
	\end{subfigure}
	\caption{Behavior of path loss function in (\ref{eq:f(r)}) when varying $d$.}
	\label{fig:app2}
\end{figure}

We also illustrate a situation where $d$ is varied holding $M$ constant in Fig. \ref{fig:app2}, for $d$ between $40$ and $70$ in steps of $5$, and a small and large value of $M$. It is clear that for the small-$M$ case, $r_{max}$ begins to decrease between $d = 55$ and $60$, despite the two other real roots not having occurred. The curve for $d = 60$ is shown in bold to illustrate the reduction in $r_{max}$. In the large-$M$ case, $r_{max}$ continually increases as the other roots appear.

These observations are sufficient to imply the existence of two regimes depending on $M$, and that there are different expressions for $d^{*}$ under each regime. The boundary can be thought of as a differentiator between two cases as $d$ is increased:
\begin{itemize}
\item \textbf{Case (a)}: Either $r_{max}$ stops increasing before the occurrence of a real double root in $f(r)$ (which includes the case where $M \leq 9$, where the double roots are complex), in which case the largest feasible $r$, i.e., $r_{max}$, determines $d^{*}$; or
\item \textbf{Case (b)}: $r_{max}$ stops increasing after the occurrence of a real double root in $f(r)$, in which case $d^{*}$ is equal to the smallest value of $d$ where the double root first occurs.
\end{itemize}
These are the only two possibilities; therefore, such a boundary exists.

Ideally, to obtain $d^{*}$, one first determines the smallest positive real root of (\ref{eq:f(r)}), and then finds $d$ which maximizes the root expression. However, due to the unwieldy form of the exact roots of quartic polynomials, we take an alternative approach. In \textbf{Case (a)}, $r_{max}$ is simply the largest feasible value of $r$, regardless of whether a coverage gap exists or not. Therefore, we determine $d^{*}$ by considering (\ref{eq:f(r)}) as a function of $d$, which is a convex quadratic function:
\begin{equation}
f(d) = r^{2} d^{2} - 2 r^{3} \Theta d + (r^{4} - \varsigma). 	\label{eq:f(d)}
\end{equation}
For a given $r$, the existence of real roots to $f(d)$ effectively indicates that the value of $r$ is feasible. In other words, if a BD lies on the sector edge at distance $r$ from the reader, then there exists a range of $d$ within which the equivalent SNR constraint can be satisfied for that BD; that range of $d$ is between the two roots of (\ref{eq:f(d)}). Using this observation, we first find $r_{max}$ by using the fact that $f(d)$ has a double root when $r = r_{max}$; that is, the double root is the point where only one value of $d$ makes $r$ just feasible. By completing the square on (\ref{eq:f(d)}), we require
\begin{equation}
r_{max}^4 - \varsigma - \frac{4 r_{max}^{6} \Theta^{2}}{4 r_{max}^{2}} = 0 \hspace{2mm} \Rightarrow \hspace{2mm} r_{max} = \left( \frac{\varsigma}{1 - \Theta^{2}} \right)^{\frac{1}{4}}.		\label{appC_1}
\end{equation}
The corresponding solution for $d$ can then be found by computing the vertex of the quadratic in (\ref{eq:f(d)}), which equates to $\frac{2 r_{max}^{3} \Theta}{2 r_{max}^{2}} = r_{max} \Theta$. 

From the previous discussion on \textbf{Case (a)}, we note that $r_{max}$ in (\ref{appC_1}) is valid up to some threshold for $M$, denoted by $M_{th}$, where $M_{th} > 9$. In other words, $M_{th}$ is the point where $r_{cov}$ switches from the largest positive real root to the smallest positive real root of (\ref{eq:f(r)}). This gives the first part to Theorem 1, and is summarized in Lemma 2 below.

\begin{lemma}
For $M \leq M_{th}$, the optimal PB deployment distance $d^{*}$ is given by
\begin{equation}
d^{*} = \left( \frac{\varsigma}{1 - \Theta^{2}} \right)^{\frac{1}{4}} \Theta.  \label{eq:thm1a}
\end{equation}
\end{lemma}

Note that as $M \rightarrow \infty$, $d^{*}$ in (\ref{eq:thm1a}) approaches infinity, which is not possible even with large $M$ when a single serving PB is assumed. This confirms the existence of a boundary between the small-$M$ regime and the large-$M$ regime. Now we are required to determine $M_{th}$. In light of this, we examine \textbf{Case (b)}, where we revisit the case where (\ref{eq:f'(r)}) takes on a double root, which only becomes possible when $M \geq 10$. In other words, there exists a point where as $d$ increases, (\ref{eq:f(r)}) switches from having one positive real root to first a single root plus a double root, and then three positive real roots. Taking advantage of this behavior, we provide the following closed-form result for the value of $d$ at which the double root occurs.

\begin{lemma}
For $M \geq 10$, the distance $d^{'}$ at which (\ref{eq:f(r)}) admits a real double root is given by
\begin{multline}
d^{'} = \bigg( -\frac{1}{2} \csc^{2}\left( \frac{\pi}{M} \right) \bigg( \sqrt{\varsigma^{2} \Theta^{2} \left( 9 \Theta^{2} - 8 \right)^{3}} \\ + \varsigma \left( 27 \Theta^{4} - 36 \Theta^{2} + 8 \right) \bigg) \bigg)^{\frac{1}{4}}.	\label{eq:thm1}
\end{multline}
\end{lemma}

\begin{proof}
See Appendix B.
\end{proof}

It can be shown that (\ref{eq:thm1}) converges to a finite value as $M \rightarrow \infty$. As such, we are now required to find the intersection of the two expressions in Lemmas 2 and 3. This allows us to determine $M_{th}$ where the correct $d^{*}$ switches from (\ref{eq:thm1a}) to (\ref{eq:thm1}); that is, the value of $d$ where the double root in (\ref{eq:f(r)}) occurs before $r_{max}$ stops increasing. The value of $M_{th}$ can be obtained using standard mathematical packages. It follows that the correct value of $d^{*}$ for a given $M$ is the smaller of the expressions in Lemmas 1 and 3. This completes the proof.


\section{Proof of Lemma 3}
We derive the closed-form expression of $d^{'}$, which holds when (\ref{eq:f(r)}) admits a real double root. By rearranging the coefficients of the depressed quartic equation in \cite{quartic}, for a quartic equation of the form $a_{4} x^{4} + a_{3} x^{3} + a_{2} x^{2} + a_{1} x + a_{0} = 0$, the conditions for a double root to occur are:
\begin{align}
\Delta &= 0, \label{appB_2a} \\
8 a_{4} a_{2} - 3 a_{3}^{2} &< 0, \label{appB_2b} \\
64 a_{4}^{3} a_{0} - 16 a_{4}^{2} a_{3} a_{1} - 16 a_{4}^{2} a_{2}^{2} + 16 a_{4} a_{3}^{2} a_{2} - 3 a_{3}^{4} &< 0, \label{appB_2c} \\ 
a_{2}^{2} + 12 a_{4} a_{0} - 3 a_{3} a_{1} &\neq 0, 	\label{appB_2d}
\end{align}
where $\Delta$ is the discriminant. When $a_{1} = 0$, 
\begin{align}
\Delta &= 256 a_{4}^{3} a_{0}^{3} - 128 a_{4}^{2} a_{2}^{2} a_{0}^{2} + 144 a_{4} a_{3}^{2} a_{2} a_{0}^{2} + 16 a_{4} a_{2}^{4} a_{0} \nonumber \\
		& \quad - 27 a_{3}^{4} a_{0}^{2} - 4 a_{3}^{2} a_{2}^{3} a_{0} \nonumber \\
		&= \left( 16 \varsigma \Theta^{2}\!-\!16 \varsigma \right) d^{8}\!+\!\left( 576 \varsigma^{2} \Theta^{2}\!-\!128 \varsigma^{2}\!-\!432 \varsigma^{2} \Theta^{4} \right) d^{4} \nonumber \\ 
		& \quad - 256 \varsigma^{3}.	\label{appB_3}
\end{align}

Equation (\ref{appB_3}) is quadratic with respect to $d^{4}$. Solving (\ref{appB_3}) and taking the positive root yields the result in (\ref{eq:thm1}). The negative root is given by
\begin{multline}
(d^{4})_{-} = \left(\!\frac{1}{4}\!\left( 9 \varsigma\!+\!27 \varsigma \cos \left( \frac{2 \pi}{M} \right) \right. \right. \\ \left. \left. + 2 \left( \varsigma\!+\!\sqrt{\varsigma^{2} \Theta^{2} \left( 9 \Theta^{2}\!-\!8 \right)^{3} } \right) \csc^{2}\!\left( \frac{\pi}{M} \right) \right)\!\right)^{\frac{1}{4}}.	\label{negativeRoot}
\end{multline}
$(d^{4})_{-}$ tends to infinity as $M$ grows large, and therefore cannot be the correct solution. It can be shown that (\ref{negativeRoot}) upper bounds (\ref{eq:thm1a}) for $M \geq 10$, implying that no boundary exists between the small-$M$ and large-$M$ regimes, contradicting previous observations. We check that the other conditions are satisfied for the positive root. Inequality (\ref{appB_2b}) is independent of $d$ and simplifies to $M \geq 6$, which is already true since it is assumed that $M \geq 10$. Rearranging (\ref{appB_2c}) gives
\begin{equation}
d^{4} > \frac{4 \varsigma}{1 + 3 \Theta^{4} - 4 \Theta^{2}}, \label{appB_4}
\end{equation}
where the RHS is always negative. Inequation (\ref{appB_2d}) simplifies to $d^{4} \neq 12 \varsigma$, which is the condition required for a triple root to occur. However, a triple root can only occur when the discriminant in Lemma 1 is zero, which is not possible for an integer value of $M$. Therefore, the positive root in (\ref{eq:thm1}) is the correct solution, as long as (\ref{appB_2c}) and (\ref{appB_2d}) are not violated.


\bibliographystyle{ieeetran}
\bibliography{IEEEabrv,biplace_ref}

\end{document}